\documentclass[dvipsnames]{acmart}

\usepackage{hyperref}
\usepackage{cleveref}
\crefname{theorem}{thm.}{thm.}
\crefname{lemma}{lem.}{lem.}
\crefname{corollary}{cor.}{cor.}
\crefname{section}{sec.}{sec.}
\crefname{equation}{}{}
\crefname{conjecture}{conj.}{conj.}

\usepackage{mathtools}
\usepackage{todonotes}
\usepackage{tikz}
\usepackage{bm}
\usepackage{thmtools}

\usepackage{caption}
\usepackage{subcaption}

\newcommand{\powersofk}[1]{{#1}^{\mathbf{N}}}
\newcommand{\kthpowers}[1]{{\mathbf{N}}^{#1}}

\newcommand{\tuple}[1]{\langle{#1}\rangle}
\newcommand{\floor}[1]{\lfloor{#1}\rfloor}

\newcommand{\Acc}[1]{\mathsf{Acc}_{#1}}

\newcommand{\automaton}{\mathcal{A}}

\newcommand{\torus}{\mathbb{T}}
\newcommand{\qinit}{q_{\operatorname{init}}}
\newcommand{\rinit}{r_{\operatorname{init}}}

\newcommand{\journey}{\mathsf{jour}}

\newcommand{\naturals}{\mathbb{N}}

\newcommand{\nat}{\mathbb{N}}
\newcommand{\intg}{\mathbb{Z}}
\newcommand{\rel}{\mathbb{R}}
\newcommand{\rat}{\mathbb{Q}}
\newcommand{\com}{\mathbb{C}}
\newcommand{\alg}{\overline{\rat}}
\newcommand{\ralg}{\rel \cap \alg}

\newcommand{\Acal}{\mathcal{A}}
\newcommand{\Bcal}{\mathcal{B}}
\newcommand{\Ccal}{\mathcal{C}}

\newcommand{\Fcal}{\mathcal{F}}

\newcommand{\Lcal}{\mathcal{L}}
\newcommand{\Mcal}{\mathcal{M}}

\newcommand{\Tcal}{\mathcal{T}}

\newcommand{\logtwo}{0.693}

\newcommand{\logthree}{1.0986}

\newcommand{\logten}{2.302}

\newcommand{\seq}[1]{\langle #1 \rangle_{n=0}^{\infty}}
\newcommand{\st}{\colon}

\newcommand{\zerovec}{\mathbf{0}}
\newcommand{\rexp}{\mathbb{R}_{\exp}}

\newcommand{\characteristic}[1]{\operatorname{Char}(#1)}
\newcommand{\order}[1]{\operatorname{Ord}(#1)}

\newtheorem{problem}{Problem}
\newtheorem*{theorem*}{Theorem}
\newtheorem*{corollary*}{Corollary}

\begin{document}

	\title{On the Decidability of Monadic Theories of Arithmetic Predicates}


\author{Val\'erie Berth\'e}
\email{berthe@irif.fr}
\affiliation{%
\institution{Universit\'e Paris Cité, IRIF, CNRS}
\streetaddress{}
\city{Paris}
\country{France}
\postcode{75013}}
\orcid{0000-0001-5561-7882}
\author{Toghrul Karimov}
\email{toghs@mpi-sws.org}
\orcid{0000-0002-9405-2332}
\author{Joris Nieuwveld}
\email{jnieuwve@mpi-sws.org}
\orcid{0009-0002-0339-1230}
\author{Jo\"el Ouaknine}
\email{joel@mpi-sws.org}
\orcid{0000-0003-0031-9356}
\author{Mihir Vahanwala}
\email{mvahanwa@mpi-sws.org}
\orcid{0009-0008-5709-899X}
\affiliation{%
\institution{Max Planck Institute for Software Systems}
\streetaddress{Saarland Informatics Campus}
\city{Saarbr\"ucken}
\country{Germany}
\postcode{66123}
}
\author{James Worrell}
\email{jbw@cs.ox.ac.uk}
\orcid{0000-0001-8151-2443}
\affiliation{%
\institution{University of Oxford, Department of Computer Science}
\streetaddress{}
\city{Oxford}
\country{United Kingdom}
\postcode{OX1 3QD}
}
\renewcommand{\shortauthors}{Berth\'e \emph{et al.}}

\begin{abstract}
We investigate the decidability of the monadic second-order (MSO) theory of the structure $\langle \mathbb{N};<,P_1, \ldots,P_d \rangle$, for various unary
predicates $P_1,\ldots,P_d \subseteq \mathbb{N}$. We focus in particular on `arithmetic' predicates arising in the study of linear recurrence sequences, such
as fixed-base powers $\powersofk{k} = \{k^n : n \in \mathbb{N}\}$,
$k$th powers $\kthpowers{k} = \{n^k : n \in \mathbb{N}\}$, and the set of
terms of the Fibonacci sequence $\mathsf{Fib} =
\{0,1,2,3,5,8,13,\ldots\}$ (and similarly for other linear recurrence
sequences having a single, non-repeated, dominant characteristic root).
We obtain several new unconditional and conditional decidability results, a
select sample of which are the following:
\begin{itemize}
\item The MSO theory of
$\langle \mathbb{N};<,\powersofk{2}, \mathsf{Fib} \rangle$ is decidable;
\item The MSO theory of
$\langle \mathbb{N};<, \powersofk{2}, \powersofk{3}, \powersofk{6}
\rangle$ is decidable;
\item The MSO theory of
$\langle \mathbb{N};<, \powersofk{2}, \powersofk{3}, \powersofk{5}
\rangle$ is decidable assuming Schanuel's conjecture;
\item  The MSO theory of
$\langle \mathbb{N};<, \powersofk{4}, \kthpowers{2}\rangle$ is
decidable;
\item  The MSO theory of
$\langle \mathbb{N};<, \powersofk{2}, \kthpowers{2} \rangle$ is
Turing-equivalent
to the MSO theory of $\langle \mathbb{N};<,\gamma \rangle$, where $\gamma \st \nat \to \{0,1\}$ is the binary expansion of
$\sqrt{2}-1$. The widely believed conjecture that $\sqrt{2}$ is \emph{normal} implies decidability of both MSO theories.
\end{itemize}
These results are obtained by exploiting and combining techniques from
dynamical systems, number theory, and automata theory.
This paper is the journal version of \cite{berthe2024fullversion}.
\end{abstract}

	



	\maketitle

\section{Introduction}
\label{introduction}

B\"uchi's seminal 1962 paper~\cite{buchi-MSO} established the
decidability of the monadic second-order (MSO) theory of the structure
$\langle \mathbb{N};<\rangle$,
and in so doing brought to light the profound connections between
mathematical logic and automata theory. Over the ensuing decades,
considerable work has been devoted to the question of which expansions
of $\langle \mathbb{N};<\rangle$ retain MSO decidability. In other
words, for which unary predicates $P_1,\ldots,P_d$ is the MSO theory of
$\langle \mathbb{N};<,P_1,\ldots,P_d\rangle$ decidable?\footnote{The restricted focus on \emph{unary} (or \emph{monadic}) predicates is justified by the fact that most natural non-unary predicates immediately lead to undecidability; see, e.g., \cite[Thm.~3]{thomas1975note}.}
 Here by unary predicate we mean
a fixed set of non-negative integers $P \subseteq \mathbb{N}$. Taking,
for example, $P$ to be the set of prime numbers, B\"uchi and
Landweber~\cite{buchi1969definability}
observed in 1969 that a proof of decidability of the MSO theory of
$\langle \mathbb{N};<,P\rangle$ would ``seem very difficult'', as it
would \emph{inter alia} enable one (at least in principle) to settle the twin prime
conjecture. (Decidability was subsequently established assuming the linear case of Schinzel's hypothesis H~\cite{bateman1993decidability}, also known as Dickson's conjecture.)

The set of prime numbers is, of course, highly intricate. In 1966,
Elgot and Rabin~\cite{elgotrabin} considered
a large class of simpler predicates of `arithmetic' origin, such as, for
any fixed $k$, the set $\powersofk{k} = \{k^n : n \in \mathbb{N}\}$
of powers of $k$, and the set $\kthpowers{k} = \{n^k : n \in \mathbb{N}\}$ of $k$th powers. For any such predicate $P$,
they systematically established decidability of the MSO theory of
$\langle \mathbb{N};<,P\rangle$ by using the so-called \emph{contraction method}. 
Many years later, their automata-theoretic results were  substantially developed and extended by, among others, Carton and
Thomas~\cite{cartonthomas}, Rabinovich~\cite{rabinovich},
and Rabinovich and Thomas~\cite{rabinovich2006decidable}, using the framework of \emph{effectively profinite ultimate periodicity}.
A related concept that plays a crucial role in this paper\footnote{More precisely, we make an extensive use of \emph{uniform recurrence} (see \Cref{sec::unif-rec-words}), which is a special case of almost periodicity.} is that of \emph{effective almost-periodicity}, introduced in the 1980s by
Sem\"enov~\cite{semenov1980related,semenov84_logic_theor_one_place_funct}, and recently
brought to bear in the MSO model checking of linear dynamical
systems~\cite{POPL22}.

It is notable that whilst Elgot and Rabin established separately the
decidability of the MSO theories, for example, of
$\langle \mathbb{N};<, \powersofk{2}\rangle$ and
$\langle \mathbb{N};<, \powersofk{3}\rangle$, they remained resolutely
silent on the obvious joint expansion
$\langle \mathbb{N};<, \powersofk{2}, \powersofk{3}\rangle$.
This in hindsight is wholly unsurprising: there are various statements that
one can express in the above theory whose truth values are highly
non-trivial to determine. 
An example is, for given fixed $a,b$, the assertion that there exist infinitely many powers of $3$ whose distance to the next power of~$2$ is congruent to $a$ modulo $b$.
An immediate corollary of
our first main result (\Cref{powers-main-structure-parametrised}) is that the MSO theory of
$\langle \mathbb{N};<, \powersofk{2}, \powersofk{3}\rangle$
is indeed decidable. 
Although this is new, we should mention
that decidability of the \emph{first-order} theory of
$\langle \mathbb{N};<, \powersofk{2}, \powersofk{3}\rangle$ was proven (using quantifier elimination) over forty years ago by Sem\"enov~\cite{semenov1980related}.

Looking over the last several decades' worth of research work
on monadic second-order expansions of the structure $\langle \mathbb{N};<\rangle$, it is
fair to say that the bulk of the attention has focused on
the addition of a \emph{single} predicate $P$. The obvious
reason is that whilst, in general, the decidability of
single-predicate expansions of $\langle \mathbb{N};<\rangle$ can usually
be handled with automata-theoretic techniques alone, by reasoning about
individual patterns in isolation, this is not the case when multiple
predicates are at play simultaneously. Such collections of predicates
can exhibit highly complex interaction patterns, which existing
approaches are ill-equipped to handle.
In this paper, we overcome these difficulties by showing that key aspects of such interactions can be modelled in the theory of \emph{dynamical systems}.

Our approach to analysing decidability of the MSO theory of a structure $\mathbb{S} \coloneqq \langle \nat; <, P_1,\ldots,P_d\rangle$ with $d \ge 1$ predicates is as follows.
Firstly, using B\"uchi's original construction \cite{buchi-MSO}, given an MSO formula $\varphi$ we construct an automaton~$\Acal$ over infinite words such that $\varphi$ holds in $\mathbb{S}$ if and only if $\Acal$ accepts the \emph{characteristic word} $\alpha$ of $\mathbb{S}$, which records, for each predicate $P_i$, the positions $n \in \nat$ such that $n \in P_i$.
That is, the decision problem of the MSO theory of $\mathbb{S}$ is Turing-equivalent to the \emph{automaton acceptance problem} for $\alpha$, denoted $\Acc{\alpha}$.
Next, we use arithmetic properties of specific $P_1,\ldots,P_d$ to argue that $\Acc{\alpha}$ is Turing-equivalent to $\Acc{\beta}$, which is the automaton acceptance problem for the so-called \emph{order word}~$\beta$ of $\mathbb{S}$.
The order word is a compressed version of the characteristic word that keeps track of only the order in which the elements of $P_1,\ldots,P_d$ occur.
In the final step, we give dynamical systems that generate the order word $\beta$, and use their various properties to show decidability of $\Acc{\beta}$ and hence the MSO theory of $\mathbb{S}$.

In this paper, we study two kinds of structures $\langle \nat; <, P_1,\ldots,P_d\rangle$.
Firstly, we study the case of $P_i = \{u^{(i)}_n \ge 0\}$, where $\seq{u^{(i)}_n}$ is an integer \emph{linear recurrence sequence} (LRS) with a single, non-repeated dominant root.
These include geometric progressions $\seq{a \rho^n}$ for integers $a, \rho \ge 1$, as well as the Fibonacci numbers.
In this setting, the relevant dynamical systems are \emph{translations on the $d$-dimensional torus $\torus^d \coloneqq [0,1)^d$}, given by $x \mapsto x + t$ for some $t \in \torus^d$ where addition operates coordinatewise and modulo 1.
These are fundamental compact dynamical systems that have been extensively studied from the perspectives of symbolic dynamics, ergodic theory, and number theory \cite{MR2953186,fogg2002substitutions}.
Below we state specialised versions of our main results for MSO theories of \emph{value sets} of LRS with a single, non-repeated dominant root.

\begin{theorem*}[Weakened version of \Cref{powers-main-structure-parametrised}]
	Suppose we are given an MSO formula $\varphi$ and positive integers $a_1,\rho_1,\ldots,a_d,\rho_d$ such that $\frac 1 {\log(\rho_1)},\ldots,\frac 1 {\log(\rho_d)}$ are linearly independent over $\rat$.
	Then it is decidable whether $\varphi$ holds in $\langle \nat; <, P_1,\ldots,P_d \rangle$, where $P_i = \{a_i\rho_i^n \st n \in \nat\}$.
\end{theorem*}
The linear independence condition holds, for example, when $d \le 2$, or the set $\{\rho_1,\ldots,\rho_d\}$ contains at most two \emph{multiplicatively independent} elements; 
see \Cref{lem::rank-d-2-to-LI} for the precise formulation.
The latter captures, for example, the case of $\rho_1 = 2$, $\rho_2 = 3$, $\rho_3 = 6$.
The role of the linear independence condition is to ensure that a certain \emph{hypercubic billiard} dynamical system, which captures the order in which the elements of $P_1,\ldots,P_d$ occur, is non-degenerate; see  \Cref{sec::cutting-sequences}.

Now suppose we have $a_1 = a_2 = a_3 = 1$, $\rho_1 = 2$, $\rho_2 = 3$, and $\rho_3 = 5$.
It is widely believed that $1/\log(2), 1/\log(3), 1/\log(5)$ are linearly independent over $\rat$.
However, no proof is known, and hence the theorem above does not apply.
Nevertheless, we can conditionally deduce the required linear independence by invoking, for example, \emph{Schanuel's conjecture} in transcendental number theory \cite[Chap.~1.4]{waldschmidt2000}.
Schanuel's conjecture is a unifying conjecture that, among others, completely describes all \emph{polynomial} relations between logarithms of algebraic numbers, which in turn captures all \emph{linear} relations between the inverses thereof.
As another example, Schanuel's conjecture implies that for any non-zero polynomial $p(x,y)$ with rational coefficients, $p(e,\pi) \ne 0$, i.e., the constants $e$ and $\pi$ are algebraically independent.
It turns out that, for power predicates, assuming Schanuel's conjecture we can in fact prove the most general decidability result possible, handling in particular the cases where $1/\log(\rho_1),\ldots,1/\log(\rho_d)$ are linearly dependent.

\begin{theorem*}[See \Cref{thm::main-schanuel-weak}]
		Suppose we are given an MSO formula $\varphi$ and integers $a_1,\rho_1,\ldots,a_d,\rho_d \ge 1$.
		Assuming Schanuel's conjecture, it is decidable whether $\varphi$ holds in $\langle \nat; <, P_1,\ldots,P_d \rangle$, where $P_i = \{a_i\rho_i^n \st n\in\nat\}$.
\end{theorem*}

We mention that throughout this paper, all conditional decision procedures that we give rely on Schanuel's conjecture only for termination and not correctness: whenever they terminate, the output assertion as to whether $\varphi$ holds is unconditionally guaranteed to correct.

Having stated the theorems above, it is natural to ask: what if we fix the structure and allow only the formula $\varphi$ to be given as the input?
That is, for which $P_1,\ldots,P_d$ is the corresponding MSO theory decidable?
We obtain the following, somewhat surprising result..

\begin{theorem*}[See \Cref{thm::main-quirky}]
	For any integers $a_1,\rho_1,\ldots,a_d,\rho_d \ge 1$, there exists an algorithm that, given an MSO formula $\varphi$, decides whether $\varphi$ holds in $\langle \nat; <, P_1,\ldots,P_d \rangle$, where $P_i = \{a_i\rho_i^n \st n\in\nat\}$.
\end{theorem*}
The caveat is that the attendant decision procedure is \emph{non-uniform} in $a_1,\rho_1,\ldots,a_d,\rho_d$, in that it ``knows'' the ideal of all polynomial relations between $\log(a_1), \log(\rho_1),\ldots,\log(a_d),\log(\rho_d)$.

When studying the ordering in which powers of integers occur (i.e., the order word corresponding to the predicates $P_1,\ldots,P_d$ above), we discover interesting connections to word combinatorics and the problem of determining the \emph{factor complexity} (i.e., the number of distinct subwords of a given length $n$) for certain classes of infinite words.
In \Cref{sec::cutting-sequences}, we discuss how it is possible to replace roughly half of the number theory we make use of with arguments from word combinatorics and automata theory, to establish weakened versions of our main theorems.

The second setting we consider is that of $\langle \nat; <, P_1, P_2 \rangle$, where $P_1 = \{qn^d \st n \in \nat\}$ and $P_2 = \{pb^n \st n \in \nat\}$ for integers $q,b,d,b$.
In this case, the underlying dynamical systems are given by maps $T_b \st [0,1) \to [0,1)$, $T_b(x) = \{b \cdot x\}$ where $b \ge 2$ is an integer and $\{y\}$ denotes the fractional part of $y\in\rel$.
Iteratively applying $T_b$ starting from $x \in [0,1)$ generates the expansion of $x$ in base $b$.
These are also fundamental dynamical systems that go as far back as the work of R\'enyi on $\beta$-expansions \cite{renyi1957representations}.
For the predicates $P_1,P_2$ above, we give a complete result that links their MSO theory to that of base-$b$ expansions of certain algebraic numbers.\footnote{We view the base-$b$ expansion of $x \in [0,1)$ as an infinite word over $\{0,\ldots,b-1\}$, or equivalently, as a function of type $\nat \to \{0,\ldots,b-1\}$.}

\begin{theorem*}[See \Cref{thm::main poly vs exp}]
	Let $b, d \ge 2$ and $p, q \ge 1$ be integers, $P_1 = \{qn^d \st n \in \nat\}$, and $P_2 = \{pb^n \st n \in \nat\}$.
	Write $\eta = \sqrt[d]{p/q}$, $\zeta = \sqrt[d]{1/b}$, and let $\gamma_0,\ldots,\gamma_{d-1} \in \{0,\ldots,b-1\}^\omega$ be the base-$b$ expansions of $\{\eta\}, \{\eta \zeta\},\dots,\{\eta\zeta^{d-1}\}$, respectively.
	Then the MSO theories of $\langle \nat; <, P_1,P_2\rangle$ and $\langle \nat; <, \gamma_0,\ldots,\gamma_{d-1} \rangle$ are Turing-equivalent.
\end{theorem*}

Various interesting MSO decidability results follow readily from the theorem above; see \Cref{sec:MSO normal numbers} for a detailed discussion.
\begin{enumerate}
	\item[(A)] Taking $p=q=1$ and $b=d=2$, we obtain that the MSO theory of $\langle \nat; <, \kthpowers{2}, \powersofk{2}\rangle$ is Turing-equivalent to that of $\langle \nat; <, \gamma\rangle$, where $\gamma$ is the binary expansion of $\sqrt{2}-1$ viewed as function of type $\nat \to \{0,1\}$.
	\item[(B)] The MSO theory of 
	$\langle \nat; <,\powersofk{b}, \kthpowers{d}\rangle$ is decidable for any  $d \ge 2$ and $b = k^d$ for some $k \ge 2$.
	(For example, the MSO theory of $\langle \nat; <, \kthpowers{4}, \powersofk{2}\rangle$ is decidable.)
	In this case, $\eta =1$, $\zeta = 1/k$, and hence $\{\eta\}, \{\eta \zeta\},\dots,\{\eta\zeta^{d-1}\} \in \rat$.
	Recall that expansions of rational numbers in any integer base are ultimately periodic, and hence their representations as functions of type $\nat \to \{0,\ldots,b-1\}$ are MSO-definable in $\langle \nat; < \rangle$, which itself has a decidable MSO theory.
\end{enumerate}

But what do we know about the expansion of an irrational algebraic number in an integer base $b$?
For such an expansion $\alpha$, it is known, for example, that $\liminf_{n \to \infty} \frac{\pi_\alpha(n)}{n} = + \infty$
where $\pi_\alpha(n)$ denotes the number of distinct finite words of length $w$ that appear in $\alpha$ (\Cref{thm::bugeaud-adamcz}).
On the other hand, many simple results that would be subsumed by decidability of the MSO theory of an expansion in base $b$ remain elusive: for example, at the time of writing no algorithm is known that decides whether a given finite word occurs in a given expansion.
Nevertheless, expansions of irrational algebraic numbers in integer bases are widely conjectured to be \emph{normal}, and \emph{a fortiori}
\emph{disjunctive}: every finite pattern of digits should occur infinitely
often. 
As the MSO theory of any disjunctive word is decidable
(\Cref{normal-decidable}), the MSO theory of $\langle \nat; <, \kthpowers{2}, \powersofk{2}\rangle$ is decidable assuming the binary expansion of $\{\sqrt{2}\} = \sqrt{2}-1$ is disjunctive.

\paragraph{Structure of the paper} We recall necessary preliminaries from logic, automata theory, number theory, and word combinatorics in \Cref{sec::prelims}.
In \Cref{sec::automata-theory}, we develop a range of automata-theoretic tools that allow us to reduce between the (decision problems of) MSO theories of various structures.
In \Cref{sec::lrs-with-one-dominant-root}, we apply our toolbox to show how to decide MSO theories of multiple linear recurrence sequences with a single, non-repeated dominant root.
\Cref{sec:MSO normal numbers} is dedicated to studying the MSO theories of structures $\langle \nat; <, \{qn^d \st n \in \nat\}, \{qn^d \st n \in \nat\}\rangle$ through the lens of base-$b$ expansions and normal numbers.
Finally, \Cref{sec::discussion} contains a brief discussion of our results and open problems concerning decidability of MSO theories.

\section{Preliminaries}
\label{sec::prelims}

We denote by $\zerovec$ the tuple $(0,\ldots,0)$ whose dimensions will be clear from the context.
We write $\torus$ for $[0,1)$, viewed as an additive group where addition operates modulo 1.
For $x \in \rel$, we write $\{x\}$ for the fractional part $x - \lfloor x \rfloor$ of $x$.

\subsection{Words and automata}
\label{constructs}

By an alphabet $\Sigma$ we mean a finite non-empty set of letters. The sets of finite, finite non-empty, and infinite words over~$\Sigma$ are denoted $\Sigma^*, \Sigma^+$, and $\Sigma^\omega$, respectively.
For a finite or infinite word $\alpha$ and $n \in \nat$, we write $\alpha(n)$ for the $n$th letter of~$\alpha$.
Thus $\alpha = \alpha(0)\alpha(1)\cdots$.
We define $\alpha[n, m) \coloneqq \alpha(n)\cdots\alpha(m-1)$, and assuming $\alpha$ is infinite, $\alpha[n,\infty) \coloneqq \alpha(n)\alpha(n+1)\cdots$.
We denote the length of a finite word $w$ by $|w|$.
A finite word $w \in \Sigma^*$ \emph{occurs} at a position $n$ in $\alpha$ if $\alpha[n, n+|w|) = w$.
Such $w$ is called a \emph{factor} of $\alpha$.
We denote by $\pi_\alpha(n)$ the number of distinct factors of length $n$ of $\alpha$.
The function $\pi_\alpha$ is called the \emph{factor complexity} of $\alpha$.

A \emph{deterministic finite Muller automaton} (simply called an \emph{automaton} throughout the paper) $\Acal$ over an alphabet $\Sigma$ is given by a tuple $(Q, \qinit, \delta, \Fcal)$,
where $Q$ is the (finite) set of states, $\qinit \in Q$ is the initial state, $ \delta\colon Q \times \Sigma \rightarrow Q$ is the transition function, and $\Fcal \subseteq 2^Q$ is the acceptance condition.
For $q \in Q$ and $u \in \Sigma^*$, we denote by $\delta(q, u)$ the state obtained when the automaton reads $u$ starting at the state $q$.
We denote by $\Acal(\alpha)$ the sequence of states visited when $\Acal$ reads $\alpha$. A word $\alpha \in\Sigma^\omega$ \emph{is accepted} by $\Acal$ if the set~$S$ of states appearing infinitely often in $\Acal(\alpha)$ is present in~$\Fcal$.
We write $\alpha \in L(\Acal)$ to mean that $\Acal$ accepts $\alpha$.

A \emph{deterministic finite transducer} (simply called a \emph{transducer} throughout the paper)  $\Bcal$ over an input alphabet $\Sigma$ and an output alphabet $\Gamma$ is given by $(R, \rinit, \sigma)$, where $R$ is the (finite) set of states, $\rinit \in R$ is the initial state, and $\sigma\st R \times \Sigma \rightarrow R \times \Gamma^*$ is the transition function.
At every step, $\Bcal$ reads a letter from the input alphabet $\Sigma$, transitions to the next state, and outputs a finite word over the output alphabet~$\Gamma$.
We denote by $\Bcal(\alpha)$ the (possibly finite) word over~$\Gamma$ output by $\Bcal$ upon reading $\alpha \in \Sigma^\omega$.

Let $\Acal$ be an automaton as above.
By a \emph{journey} on $\Acal$ we mean an element of $J \coloneqq Q \times Q \times 2^{Q}$.
A path $q_0q_1q_2\cdots q_n \in Q^{n+1}$ \emph{makes} the journey $(q_0, q_n, V)$ where $V$ is the set of states occurring in the proper suffix $q_1q_2\cdots q_n$.
If $n\ge 1$, then $q_n \in V$ necessarily, but $q_0$ may not belong to $V$.
The unique journey a word $w \in \Sigma^*$ makes starting in $q_0 \in Q$, denoted by $\journey(w, q_0)$, is the journey made by the path $q_0\cdots q_{|w|}$, where $q_{i+1} = \delta(q_i, w(i))$ for $1 \le i < |w|$.
The empty word makes journeys of the form $(q, q, \emptyset)$.
If $v$ makes the journey $(q_1, q_3, V_1)$ and $w$ makes the journey $(q_3, q_2, V_2)$, then $vw$ makes the journey $(q_1, q_2, V_1 \cup V_2)$.

Next, we define the  equivalence relation $\sim_{\automaton}$ as follows.
Two words $v, w \in \Sigma^*$ are equivalent, denoted $v \sim_{\automaton} w$ if the sets of journeys they can undertake (starting from various states) are identical.
The equivalence is moreover a congruence: if $v \sim_{\automaton} w$ and $x \sim_{\automaton} y$, then $vx \sim_{\automaton} wy$.
Observe that ${\sim_{\Acal}}$ is not the classical congruence (see, e.g., \cite[Sec.~2]{cartonthomas}) associated with the automaton $\Acal$.
Our choice, however, will be more convenient for technical reasons.

Since there are only finitely many equivalence classes of ${\sim_{\Acal}}$, the quotient of $\Sigma^*$ by ${\sim_{\automaton}}$ is a finite monoid $M$, called the \emph{journey monoid}.
We use $h$ to denote the natural morphism from $\Sigma^*$ into $M$.
The morphism $h$ maps each letter to its equivalence class modulo~${\sim_{\automaton}}$.
We also extend the function $\journey$ to take inputs from $M \times Q$:
for an equivalence class $m = [w]$ and state $q$, we define $\journey(m, q) = \journey(w, q)$.
Finally, we will need the following lemma, whose proof is immediate.
\begin{lemma}
	\label{journeys}
	Let $\Acal$ be an automaton as above and $\alpha = u_0u_1\cdots \in \Sigma^\omega$, where $u_n \in \Sigma^*$ for all $n$.
	Then run of $\Acal$ on $\alpha$ can be decomposed as the concatenation of journeys
	\[
	(q_0, q_1, V_0)(q_1, q_2, V_1)(q_2, q_3, V_2)\cdots
	\]
	where $q_0 = \qinit$ and $\journey(u_n, q_n) = (q_n, q_{n+1}, V_n)$ for all $n$.
	Moreover, every $q \in Q$ appears infinitely often in $\Acal(\alpha)$ if and only if $q \in V_n$ for infinitely many $n \in \naturals$.
\end{lemma}

\subsection{Monadic second-order logic}\label{sec::mso}
Monadic second-order logic (MSO) is an extension of first-order logic that allows quantification over subsets of the universe.
Such subsets can be viewed as unary (that is, monadic) predicates.
We will only be interpreting MSO formulas over expansions of the structure $\langle \nat; < \rangle$.
For a general perspective on MSO, see \cite{mso-notes}.

Let $\mathbb{S} \coloneqq \langle \nat; <, P_1,\ldots,P_d\rangle$ be a structure where each $P_i \subseteq \nat$ is a unary predicate.
We associate a language $\Lcal_{\mathbb{S}}$ of terms and formulas with $\mathbb{S}$ as follows.
The terms of $\Lcal_{\mathbb{S}}$ are the countably many constant symbols $\{0,1,2,\ldots\}$, lowercase variables that stand for elements of $\nat$, uppercase variables that denote subsets of $\nat$, and the predicate symbols $P_1,\ldots,P_d$ which stand for the predicates $P_1,\ldots,P_d$, respectively.
The formulas of $\Lcal_{\mathbb{S}}$ are the well-formed statements constructed from  the terms, the built-in equality ($=$) and membership ($\in$) symbols, logical connectives, quantification over elements of $\nat$ (written $Qx$ for a quantifier $Q$), and quantification over subsets (written $QX$ for a quantifier $Q$).
The MSO theory of the structure $\mathbb{S}$ is the set of all sentences belonging to $\Lcal_{\mathbb{S}}$ that are true in $\mathbb{S}$.
We write $\mathbb{S} \models \varphi$ to mean that the formula $\varphi$ holds in the structure $\mathbb{S}$.
The MSO theory of~$\mathbb{S}$ is \emph{decidable} if there exists an algorithm that, given a sentence $\varphi \in \Lcal_{\mathbb{S}}$, decides whether $\mathbb{S} \models \varphi$.

As an example, consider $\mathbb{S} = \langle\nat;<, P\rangle$ where $P$ is the set of all primes.
Let $s(\cdot)$ be the successor function defined by $s(x)=y$ if and only if
\[
x < y \:\:\:\land\:\:\: \forall z\st (x< z \Rightarrow y \le z).
\]
That is, $s(x) = x+1$.
Further let
\begin{align*}
	\varphi(X) &\coloneqq 1 \in X \:\land\: 0,2 \notin X \:\land\: \forall x\st (x \in X \Leftrightarrow s(s(s(x))) \in X)\\
	\psi &\coloneqq \exists X\st \big(\varphi(X) \:\land\: \forall y \, \exists z \st (z > y \:\land\: z\in X \:\land\: P(z))\big).
\end{align*}
The formula $\varphi$ defines the subset ${\{n \st n \equiv 1 \, (\bmod \, 3)\}}$ of $\nat$, and $\psi$ is the sentence ``there are infinitely many primes congruent to 1 modulo 3'', which is true.
At the time of writing, it is not known whether the MSO theory of the structure~$\mathbb{S}$ above is decidable.

The \emph{automaton acceptance problem} for an infinite word $\alpha$, denoted $\Acc{\alpha}$, is to determine, given an automaton $\Acal$, whether $\alpha \in L(\Acal)$.
Let $P_1,\ldots,P_d \subseteq \nat$ be predicates and $\Sigma = \{0,1\}^d$.
\begin{definition}
	\label{def:characteristicword}
	The \emph{characteristic word} of $P_1,\ldots,P_d$, written $\alpha \coloneqq \characteristic{P_1,\ldots,P_d}\in \Sigma^\omega$, is defined by $\alpha(n) = (b_{n,1},\ldots,b_{n,d})$ where $b_{n,i} = 1$ if $n \in P_i$ and $b_{n,i} = 0$ otherwise.
\end{definition}
The following is a reformulation of the seminal result of B\"uchi  through which he showed decidability of the MSO theory of $\langle \nat; < \rangle$.
It will also be used to prove all of our MSO decidability results.
We state it for deterministic Muller automata, which are equivalent to \emph{non-deterministic B\"uchi automata} \cite{thomas1997languages} used in the original proof of B\"uchi.

\begin{theorem}[\unexpanded{\cite[Thms. 5.4 and 5.9]{thomas1997languages}}]
	\label{thm:MSOautomaton}
	Given an MSO formula $\varphi$ over predicates\footnote{Here we identify the predicate symbols $P_1,\ldots,P_d$ with the predicates $P_1,\ldots,P_d$.} $P_1,\ldots,P_d$, we can construct an automaton $\Acal$ over $\Sigma \coloneqq \{0,1\}^d$ such that for any structure $\mathbb{S} \coloneqq \langle \nat; <, P_1,\ldots,P_d\rangle$ with the characteristic word $\alpha$,
	\[
	\mathbb{S} \models \varphi \Leftrightarrow \alpha \in L(\Acal).
	\]
	In particular, the decision problem of the MSO theory of $\mathbb{S}$ is Turing-equivalent to the decision problem $\Acc{\alpha}$.
\end{theorem}

\subsection{Algebraic numbers}
\label{sec::alg-numbers}

A complex number $\lambda$ is algebraic if there exists a non-zero polynomial $p \in \rat[x]$ such that $p(\lambda) = 0$.
The set of algebraic numbers is denoted by $\alg$.
The unique irreducible monic polynomial $p \in \rat[x]$ that has $\lambda$ as a root is called the \emph{minimal polynomial} of $\lambda$.
A \emph{representation} of an algebraic number $\lambda$ consists of its minimal polynomial $p$ and sufficiently accurate rational approximations of the real and imaginary parts of $\lambda$ to distinguish it from the other roots of $p$.
All arithmetic operations can be performed effectively on representations of algebraic numbers; see, for example, \cite[Sec.~4.2]{cohen2013course} and \cite[Sec.~1.5.4]{karimov2023algorithmic}.


Let $z_1,\ldots,z_d \in \com$ be non-zero.
The free abelian group 
\[
G_M(z_1,\ldots,z_d) \coloneqq \{ (k_1,\ldots,k_d) \st z_1^{k_1}\cdots z_d^{k_d} = 1\}
\]
is called the \emph{group of multiplicative relations} of $z_1,\ldots,z_d$.
We say that $z_1,\ldots,z_d$ are \emph{multiplicatively independent} if $G \coloneqq G_M(z_1,\ldots,z_d)$ is the trivial group.
A \emph{basis} of $G$ is a set $B = \{\mathbf{v}_1,\ldots,\mathbf{v}_m\} \subseteq G$ that is linearly independent over $\intg$ with the property that every $\mathbf{z}\in G$ can be uniquely written as an integer linear combination of $\mathbf{v}_1,\ldots,\mathbf{v}_m$.
The \emph{rank} of $G$ is the size of any of its \emph{bases}.
For non-zero $z_1,\ldots,z_d \in \alg$, we can compute a basis of $G$ using a deep result of Masser~\cite{masser1988linear}, or the more recent polynomial-time algorithm of Combot \cite{combot2025computing}.


\subsection{Linear recurrence sequences}
\label{sec::lrs}

A sequence $\seq{u_n}$ over a ring $R$ is a \emph{linear recurrence sequence} (LRS) over $R$ if there exist $d \ge 0$ and $c_1,\dots,c_d \in R$ such that 
\begin{equation}\label{eq::lrs-2}
	u_{n+d} = c_1u_{n+d-1}+ \cdots +c_du_n
\end{equation}
for all $n \in \nat$.
The smallest such $d$ is called the \emph{order} of $\seq{u_n}$. 
We will mostly work with LRS over $\intg$, which we also call \emph{integer LRS}.
For example, the Fibonacci sequence satisfies $u_{n+2} = u_{n+1} + u_n$ for all $n \in \nat$, and is an integer LRS of order two.
We refer the reader to the book \cite{recseq} for a detailed account of LRS\@.

Let $R \subseteq \alg$ and $\seq{u_n}$ be an LRS over $R$ of order $d$.
Then there exist unique $c_1,\ldots,c_d \in R$ such that $\seq{u_n}$ satisfies the recurrence relation \Cref{eq::lrs-2}.
The \emph{minimal polynomial} of $\seq{u_n}$ is $p(x) = x^d - \sum_{i=1}^d c_i x^{d-i}$.
Suppose $p$ has the (distinct) roots $\lambda_1,\dots,\lambda_m \in \alg$, called the \emph{characteristic roots} of $\seq{u_n}$.
Then there exist unique non-zero polynomials $q_1,\ldots,q_m \in \alg[x]$ such that
\begin{equation}\label{eq::lrs-1}
	u_n = q_1(n)\lambda_1^n + \cdots+ q_m(n)\lambda_m^n
\end{equation}
for all $n \in \nat$.
Equation~\eqref{eq::lrs-1} is known as the \emph{exponential-polynomial form} of $\tuple{u_n}_{n=0}^\infty$.
A characteristic root $\lambda_i$ is called \emph{non-repeated} (alternatively, simple) if $q_i$ is constant.
The sequence $\seq{u_n}$ is called \emph{diagonalisable} (alternatively, \emph{simple}) if every $\lambda_i$ is non-repeated.
A characteristic root $\lambda_i$ is called \emph{dominant} if $|\lambda_i| \ge |\lambda_j|$ for all $1 \le j \le m$.
The sequence $\seq{v_n}$ given by $v_n = \sum_{i \in I} q_i(n)\lambda_i^n$, where $I = \{i \st \lambda_i \textrm{ is dominant}\}$, is called the \emph{dominant part} of $\seq{u_n}$.
Similarly, $\seq{u_n - v_n}$ is called the \emph{non-dominant part} of $\seq{u_n}$.
Both the dominant and the non-dominant parts are LRS over $\alg$.

Decision problems of linear recurrence sequences, despite being of central interest in algebraic number theory, largely remain open.
The most famous example is the \emph{Skolem Problem}, which asks to decide whether a given integer LRS $\seq{u_n}$ contains zero.
It is known to be decidable for sequences of order $d \le 4$, and is open at order 5 and above~\cite{mignotte-shorey-tijdeman-skolem}.
Similarly, the \emph{Positivity Problem} asks to decide whether a given integer LRS $\seq{u_n}$ satisfies $u_n \ge 0$ for all $n$.
Decidability (or, for that matter, undecidability) of the Positivity Problem would imply substantial new results in Diophantine approximation that are currently believed to be out of reach~\cite{ouaknine2014positivity}.
For this reason, any decision problem to which the Positivity Problem can be reduced is referred to as \emph{Positivity-hard}.
Finally, the \emph{Ultimate Positivity Problem} asks to decide whether a given integer LRS $\seq{u_n}$ satisfies $u_n \ge 0$ for all sufficiently large $n$.
This problem is known to be decidable for LRS of order at most 5 as well as LRS (of any order) whose characteristic roots are all non-repeated~\cite{ouaknine2014positivity, ouaknine2014ultimate}.
At order 6, decidability is again linked to certain open problems in Diophantine approximation \cite{ouaknine2014positivity}.

We conclude this section two straightforward lemmas about linear recurrence sequences.
First, the exponential-polynomial form \eqref{eq::lrs-1} immediately implies an exponential upper bound on $|u_n|$, formalised below.
\begin{lemma}\label{lem:growth LRS}
	Let $\tuple{u_n}_{n=0}^\infty$ be an LRS over $\alg$ and $r, L > 0$ be algebraic.
	Suppose $L > |\lambda_i|$ for any characteristic root $\lambda_i$ of $\seq{u_n}$.
	Then we can compute $N \in \naturals$ such that $|u_n| \le rL^n$ for all $n \ge N$.
\end{lemma}

From \Cref{eq::lrs-2} it readily follows that an integer LRS is ultimately periodic modulo any $m \in \nat_{\ge 1}$.
This is known as being \emph{procyclic}.

\begin{lemma}\label{lem:periodic LRS}
	Let $\tuple{u_n}_{n=0}^\infty$ be an integer LRS, $m$ be a positive integer, and define the sequence $\tuple{v_n}_{n=0}^\infty$ as $v_n = u_n \bmod m$.
	We can compute $N \ge 0$ and $p > 0$ such that $v_{n+p} = v_n$ for all $n \ge N$.
\end{lemma}

\subsection{Schanuel's conjecture}
A set $X = \{\alpha_1,\dots,\alpha_d\}$ of complex numbers is said to be \emph{algebraically independent} over $\rat$ if $p(\alpha_1,\dots,\alpha_d) \ne 0$ for any non-zero polynomial $p \in \rat[x_1,\ldots,x_d]$.
The \emph{transcendence degree} of $X$ is the size of a largest subset of $X$ that is algebraically independent over~$\rat$.
Below we state Schanuel's conjecture, a classical conjecture in transcendental number theory with far-reaching implications \cite{Lang66}.
\begin{conjecture}[Schanuel's conjecture]\label{conj:schanuel}
	If $\alpha_1,\dots,\alpha_d \in \com$ are linearly independent over $\rat$, then the transcendence degree of $\{\alpha_1,\dots,\alpha_d, e^{\alpha_1},\dots, e^{\alpha_d}\}$ is at least $d$.
\end{conjecture}

In particular, Schanuel's conjecture implies that $\log(\lambda_1),\ldots,\log(\lambda_d)$ are algebraically independent over $\rat$ for multiplicatively independent $\lambda_1,\ldots,\lambda_d \in \alg$.
We will use Schanuel's conjecture in the following way.
Consider the structure $\rexp \coloneqq \tuple{\rel; <, +, -, \cdot, \exp(\cdot), 0, 1}$ of real numbers equipped with arithmetic and (real) exponentiation.
By the first-order theory of $\rexp$ we mean the set of all well-formed first-order sentences that are true in~$\rexp$.
In \cite{MacintyreW96}, Macintyre and Wilkie show that the first-order theory of the structure $\rexp$ is decidable assuming Schanuel's conjecture.

\begin{theorem}
	Assuming Schanuel's conjecture, given a sentence $\varphi$, we can decide whether $\rexp \models \varphi$.
	Moreover, Schanuel's~conjecture is only required for termination and not correctness.
\end{theorem}

The second statement in the theorem above means that the algorithm of Macintyre and Wilkie always terminates assuming Schanuel's conjecture, and whenever it does, it is unconditionally guaranteed to correctly decide whether $\rexp \models \varphi$.

\subsection{Baker's theorem}
By an \emph{$R$-affine form}, where $R$ is a ring, we mean $h(x_1,\ldots,x_d) = a_0 + \sum_{i=1}^da_ix_i$, where $a_i \in R$ for all $0 \le i \le d$.
The following is a (rather weak) version of Baker's celebrated theorem on $\intg$-affine forms of logarithms \cite{baker-rational-sharp-version-1993}.

\begin{theorem}
	\label{thm::baker}
	Let $\rho_1, \ldots, \rho_d \in \rel_{>0}\cap \alg$.
	There exists a computable constant $C$ such that for all $b_1,\ldots,b_d \in \intg$ with  $B \coloneqq \max \{|b_1|+2, \ldots, |b_d|+2\}$ and $\Lambda \coloneqq b_1 \log(\rho_1) + \cdots + b_d \log(\rho_d)$,
	\[
	\Lambda \ne 0 \Rightarrow |\Lambda| > B^{-C}.
	\]
\end{theorem}
From Baker's theorem we derive the following.
\begin{theorem}
	\label{thm::baker-sum-of-two-powers}
	Let $\rho_1, \rho_2 \in \rel_{>1} \cap \alg$ and $c_1,c_2 \in \ralg$.
	There exists a computable constant $C$ with the following property.
	Write $E(n_1,n_2) \coloneqq c_1\rho_1^{n_1} + c_2\rho_2^{n_2}$ for $n_1, n_2 \in \nat$.
	Then
	\[
	E(n_1,n_2) \ne 0 \Rightarrow |E(n_1,n_2)| > \frac{\rho_1^{n_1}}{({n_1}+2)^{C}} \, , \frac{\rho_2^{n_2}}{(n_2+2)^{C}}
	\]
	for all $n_1, n_2 \in \nat$.
\end{theorem}
\begin{proof}
	We can assume that $c_1c_2 < 0$; otherwise the result is trivial.
	Suppose $E(n_1,n_2) \ne 0$.
	Then, for $(i,j) = (1,2)$ as well as $(i,j) = (2,1)$,
	\begin{align*}
		|E(n_1,n_2)| = |c_j| \cdot \rho_j^{n_j} \cdot
		\bigg \vert
		\frac{c_i}{c_j}\rho_i^{n_i} \rho_j^{-n_j} - 1
		\bigg \vert
		> 
		|c_j| \cdot \rho_j^{n_j} \cdot  
		\bigg| \log \bigg(\frac{c_i}{c_j}\bigg) + n_i \log(\rho_i) - n_j \log(\rho_j)
		\bigg|
	\end{align*}
	where we have used the fact that $e^x - 1 > x$ for all $x > 0$.
	Therefore,
	\[
	|E(n_1,n_2)| > |c_j| \cdot \rho_j^{n_j} \cdot |\Lambda(n_1,n_2)|\,,
	\]
	where $\Lambda(n_1, n_2) = |\log(c_1/c_2) + n_1\log(\rho_1) - n_2\log(\rho_2)| = |\log(c_2/c_1) + n_2\log(\rho_2) - n_1\log(\rho_1)|$.
	Therefore, we need to construct $C$ such that
	\[
	|\Lambda(n_1,n_2)| > \frac{1}{|c_1| (n_1+2)^C} \, , \frac{1}{|c_2| (n_2+2)^C}.
	\]	
	Compute $\kappa > 1$ such that for all $n_1,n_2 \in \nat_{\ge1}$,
	\begin{itemize}
		\item $n_1 > \kappa n_2 \Rightarrow \log(c_1/c_2) + n_2 \log(\rho_2) < \frac{1}{2} n_1 \log(\rho_1)$, which, in turn, implies that $|\Lambda(n_1,n_2)| > \frac{1}{2} n_1 \log(\rho_1)$ as well as $|\Lambda(n_1,n_2)| > \frac{1}{2} n_2 \log(\rho_2)$, and
		\item $n_2 > \kappa n_1 \Rightarrow \log(c_1/c_2) + n_1 \log(\rho_1) < \frac{1}{2} n_2 \log(\rho_2)$, which implies that $|\Lambda(n_1,n_2)| > \frac{1}{2} n_2 \log(\rho_2), \frac{1}{2} n_1 \log(\rho_1)$.
	\end{itemize}
	It remains to handle the values $n_1,n_2\in\nat$ such that $\frac{n_2}{\kappa} \le n_1 \le \kappa n_2$ or $n_1n_2 = 0$.
	In the former case, by Baker's theorem, there exists a computable constant $D$ such that
	\[
	|\Lambda(n_1,n_2)| > (\max \{2+n_1, 2+n_2\})^{-D} \ge (2 + \kappa n_i)^{-D}
	\]
	for $i = 1,2$.
    For the latter case, compute $M_1 := \min_{n \in \nat}\{|\Lambda(0, n)| \st \Lambda(0, n) \ne 0\}$ and $M_2 := \min_{n \in \nat}\{|\Lambda(n, 0)| \st \Lambda(n, 0) \ne 0\}$.
	It remains to choose $C$ sufficiently large such that 
	\[
	\frac{1}{(2+\kappa n)^D}, \frac{1}{2} n \log(\rho_i), M_1, M_2 > \frac{1}{|c_i| (n+2)^C}
	\]
	for all $n \in \nat$ and $i = 1,2$.
\end{proof}

Baker also proved the following theorem, which is a very special case of Schanuel's conjecture.
We use it to prove the two lemmas given below.

\begin{theorem}[Thm.~1.6 in \unexpanded{\cite{waldschmidt2000}}]
	\label{thm:Baker}
	Suppose $\rho_1,\dots,\rho_d \in \rel_{>0} \cap \alg$ are multiplicatively independent, i.e., $\log(\rho_1),\ldots,\log(\rho_d)$ are linearly independent over $\rat$.
	Then $1,\log(\lambda_1),\dots,\log(\lambda_d)$ are linearly independent over $\overline{\mathbb{Q}}$.
\end{theorem}
\begin{lemma}
	\label{lem::rank-d-2-to-LI}
	Let $d \ge 2$, $\lambda_1,\dots,\lambda_d \in \rel_{>1} \cap \alg$ be pairwise multiplicatively independent, and suppose
	\[
	\operatorname{rank}(G_M(\lambda_1,\ldots,\lambda_d)) \ge d - 2.
	\]
	Then $1/\log(\lambda_1),\dots,1/\log(\lambda_d)$ are linearly independent over $\rat$.
\end{lemma}

Note that the condition $\operatorname{rank}(G_M(\rho_1,\ldots,\rho_d)) \ge d-2$ is exactly the same as ``each triple of distinct $\rho_i, \rho_j, \rho_k$ is multiplicatively dependent'', as well as ``each triple of distinct $\log(\rho_i), \log(\rho_j), \log(\rho_k)$ is linearly dependent over $\rat$''.
This holds trivially when $d = 2$.

\begin{proof}[Proof of \Cref{lem::rank-d-2-to-LI}]
	By the two assumptions, for any distinct $i,j,k$ there exist $b_i,b_j,b_k \in \intg_{\ne 0}$ such that $\lambda_i^{b_i}\lambda_j^{b_j}\lambda_k^{b_k} = 1$.
	Equivalently, $b_i\log(\lambda_i) + b_j\log(\lambda_j) + b_k \log(\lambda_k) = 0$.
	For $1 \le j \le d$, let $b_{1,j}, b_{2,j} \in \rat$ be such that $\log(\lambda_j) = b_{1,j} \log(\lambda_1) + b_{2,j} \log(\lambda_2)$.
	Then $b_{1, 1} = b_{2, 2} = 0$ and all other $b_{i, j}$ are non-zero.

	Let $c_1,\dots,c_d \in \rat$ be such that  $\sum_{j=1}^d \frac{c_j}{\log(\lambda_j)} = 0$.
	We will argue that $c_j = 0$ for all $j$.
	Define $f(x,y) = \sum_{j=1}^d \frac{c_j}{b_{1,j} x + b_{2, j} y}$.
	Then $f(\log(\lambda_1),\log(\lambda_2)) = 0$.
	Let $g(x,y) = \prod_{j=1}^d (b_{1,j} x + b_{2, j} y)$ and
	\[
	h(x,y) 
	\coloneqq 
	f(x,y) \cdot g(x,y)
	=
	\sum_{j=1}^d  c_j\prod_{1 \le i \ne j \le d} (b_{1,i}x + b_{2,i}y)
	\]
	which simplifies to
	\[
	h(x,y) = \sum_{i=0}^{d-1} e_i x^i y^{d-i}
	\]
	for some $e_i \in \rat$.
	We have that $h(\log(\lambda_1),\log(\lambda_2)) = 0$.
	
	Suppose, for the sake of a contradiction, that $h$ is not identically zero. Then dividing by $y^{d-1}$ gives that $\log(\lambda_1)/\log(\lambda_2)$ is the root of the non-zero polynomial $\sum_{i=0}^{d-1} e_i x^i \in \rat[x]$.
	That is, $\log(\lambda_1)/\log(\lambda_2)$ is an algebraic number, say $\alpha$, and hence $\log(\lambda_1) - \alpha\log(\lambda_2) = 0$, contradicting \Cref{thm:Baker} since $\lambda_1$ and $\lambda_2$ are multiplicatively independent.
	Therefore, all $e_i$ have to be zero.
	It follows that $h$ is zero and $f$ is zero everywhere it is defined.
	
	Now suppose, for the sake of a contradiction, that $c_i \ne 0$ for some $i$.
	By the multiplicative independence assumption, the following holds.
	For every $\mu \in \rat$ and $j \ne i$, it \emph{is not} the case that for all $x,y \in \rel$, $b_{1,i}x + b_{2,i}y = \mu (b_{1,j}x + b_{2,j}y)$.
	Therefore, we can find $\widetilde{x}, \widetilde{y} \in \rat$ such that $b_{1,i}\widetilde{x} + b_{2,i}\widetilde{y} = 0$ and $b_{1,j}\widetilde{x} + b_{2,j}\widetilde{y} \ne 0$ for all $j \ne i$.
	Now consider what happens as $(x,y) \to (\widetilde{x}, \widetilde{y})$.
	We have that for all $j \ne i$, $b_{1,j}\widetilde{x} + b_{2,j}\widetilde{y} \to z_j$ for some non-zero $z_j \in \rat$, but $b_{1,i}\widetilde{x} + b_{2,i}\widetilde{y} \to 0$, which implies that $|f(x,y)|$ must diverge.
	This contradicts the zeroness of $f(x,y)$ wherever it is defined.
\end{proof}

\begin{lemma}
	\label{lem::baker-determining-sign}
	Given $\lambda_1,\ldots,\lambda_d \in \alg$ and $a_0, \ldots, a_d \in \rat$, we can effectively determine the sign of $a_0 + \sum_{i=1}^d a_i \log(\lambda_i)$.
\end{lemma}
\begin{proof}
	By computing a basis of $G_M(\lambda_1,\ldots,\lambda_d)$ (see \Cref{sec::alg-numbers}), we can rewrite this expression as $a_0 + \sum_{i=1}^e b_i \log(\mu_i)$, where for all $1 \le i \le e$, $b_i \ne 0$, $\mu_i = \lambda_j$ for some $j$, and $\mu_1,\ldots,\mu_e$ are multiplicatively independent.
    By~\Cref{thm:Baker}, this expression is zero if and only if $a_0 = b_1 = \cdots = b_e = 0$.
    In case it is non-zero, we can compute its sign by approximating its value from above and below to arbitrary precision.
\end{proof}

\subsection{Kronecker's theorem and translations on a torus}
\label{sec::toric-words}

Recall that $\torus$ denotes $[0,1)$, viewed as a group with addition modulo 1.
In our analysis of MSO theories of powers of integers, we will frequently encounter compact dynamical systems with state space $\torus^d$ and the update function $g_\delta \st \torus^d \to \torus^d$ where
\[
g_\delta((x_1,\ldots,x_d)) = (\{x_1 + \delta_1\},\ldots,\{x_d+\delta_d\})
\]
for some $\delta = (\delta_1,\ldots,\delta_d) \in \torus^d$.
These dynamical systems are well-understood thanks to Kronecker's theorem on Diophantine approximation.
Let $\delta = (\delta_1,\ldots,\delta_d) \in \torus^d$, and define the groups of affine relations and linear relations of $\delta$, respectively, as
\begin{align*}
	G_A(\delta_1,\ldots,\delta_d) &= \{(k_0, \ldots,k_d) \in \intg^{d+1} \st k_1\delta_1 + \cdots + k_d\delta_d = k_0\},\\
	G_L(\delta_1,\ldots,\delta_d) &= \{(k_1, \ldots,k_d) \in \intg^{d} \st k_1\delta_1 + \cdots + k_d\delta_d = 0\}.
\end{align*}
Further, let
\begin{equation*}
    X_\delta = \{(x_1,\ldots,x_d) \in \torus^d \st G_A(x_1,\ldots,x_d) \supseteq  G_A(\delta_1,\ldots,\delta_d)\}.
\end{equation*}
The set $X_\delta$ is closed by construction.
Because $G_A(\delta_1,\ldots,\delta_d)$ has a finite basis, $X_\delta$ can be defined by a conjunction of affine equalities with integer coefficients.
The following is a rephrasing of Kronecker's theorem in Diophantine approximation~\cite{gonek-kronecker}.
\begin{theorem}\label{thm::kronecker}
	The orbit of $\zerovec \in \torus^d$ under $g_\delta$ is dense in $X_\delta$.
\end{theorem}
That is, every open subset of $X_\delta$ is visited infinitely often by the sequence $\langle \zerovec, g_\delta(\zerovec), g_\delta(g_\delta(\zerovec)), \ldots \rangle$.
For $s \in \torus^d$, define $X_{\delta, s} = \{s+x \st x \in X_\delta\}$. 
The following is an immediate consequence of Kronecker's theorem.
\begin{theorem}\label{thm::kronecker-torus}
	The orbit of $s \in \torus^d$ under $g_\delta$ is dense in $X_{\delta,s}$.
\end{theorem}
\begin{corollary}
	\label{thm::ur-on-torus}
	Let $Z \subseteq \torus^d$ be open.
	For any $s \in \torus^d$, the set $T \coloneqq \{n \in \nat \st g_\delta^{(n)}(s) \in Z\}$ is either empty, or infinite and with uniformly bounded gaps.
\end{corollary}
\begin{proof}
	Let $\widetilde{Z} = Z \cap X_{\delta,s}$.
	If $\widetilde{Z} = \varnothing$, then $T$ is empty.
	Now suppose $\widetilde{Z} \ne \varnothing$.
	For $k \ge 1$, let $\widetilde{Z}_k = g_\delta^{(-k)}(\widetilde{Z})$, which is open and non-empty.
	By \Cref{thm::kronecker-torus}, these sets form an open cover of $\widetilde{Z}$.
	Because $X_{\delta,s}$ is compact (since it is closed and bounded), there exists $K$ such that $\widetilde{Z}_{0}, \ldots, \widetilde{Z}_{-K}$ covers $X_{\delta,s}$.
	It follows that for any $x \in X_{\delta,s}$, there exists $1 \le k \le K$ such that $g^{(k)}_\delta(x) \in \widetilde{Z}$.
	Therefore, for any $n \ge K$ there exist $n_1,n_2 \in T$ such that $n - K \le n_1 < n < n_2 \le n+K$.
\end{proof}

\subsection{Normal numbers}
Call a number $a \in \rel$ is \emph{disjunctive in base $b$} if its base-$b$ expansion $\alpha$ is a disjunctive word, i.e.\ $\alpha$ contains infinitely many occurrences of every $w \in \{0,\ldots,b-1\}^*$.
Disjunctivity is a weaker property than the more well-known property of being a \emph{normal number}: in the latter case, the \emph{frequency} of occurrences of every finite word $w \in \{0,\ldots,b-1\}$ has to be equal to $b^{-|w|}$. 
For a detailed discussion of disjunctivity, see the book~\cite{bugeaud2012distribution} and for normal numbers, see the surveys~\cite{harman2002one, queffelec2006old} and the book~\cite{MR2953186}.
In particular, \cite{harman2002one} states the following conjecture.
\begin{conjecture}\label{conj: normal}
	A real irrational algebraic number $\alpha$ is normal in any integer base $b \ge 2$.
\end{conjecture}
Thus, in particular, this conjecture implies that every real irrational algebraic number is disjunctive in any integer base.
The strongest result towards this conjecture is due to Adamczewski and Bugeaud~\cite{adamczewski2007complexity}.
Recall that $\pi_\alpha$ denotes the factor complexity of an infinite word $\alpha$.
\begin{theorem}\label{thm::bugeaud-adamcz}
	If $b \ge 2$ and $\alpha$ is the base-$b$ expansion of a real irrational algebraic number, then
	\begin{equation*}
		\liminf_{n \to \infty} \frac{\pi_\alpha(n)}{n} = + \infty.
	\end{equation*}
\end{theorem}
\subsection{Fourier-Motzkin elimination}
\label{sec::fourier-motzkin}

Fourier-Motzkin elimination is a method for solving systems of affine equalities and inequalities over the reals. 
Suppose we are given the system
\[
\Phi(x_1,\ldots,x_m) \coloneqq \bigwedge_{i \in I} a_{i,1}x_1 + \cdots + a_{i,m}x_m \sim_i b_i \:\land\: \bigwedge_{j \in J} a_{j,1}x_1 + \cdots + a_{j,m}x_m \sim_j b_j
\]
where $a_{i,k}, a_{j,k}, b_i, b_j \in \rel$ are constants, $a_{i,1}, a_{j,1} > 0$, ${\sim_i} \in \{>, \ge\}$, and ${\sim_j} \in \{<, \le\}$ for all $i \in I, j \in J, 1 \le k \le m$.
Then $\Phi(x_1,\ldots,x_m)$ can be written as
\[
\bigwedge_{i \in I} x_1 \sim_i \frac{b_i}{a_{i,1}} - \frac{a_{i,2}}{a_{i,1}} x_2 - \cdots - \frac{a_{i,m}}{a_{i,1}}x_m
\:\land\:
 \bigwedge_{j \in J} x_1 \sim_j \frac{b_j}{a_{j,1}} - \frac{a_{j,2}}{a_{j,1}} x_2 - \cdots - \frac{a_{j,m}}{a_{j,1}}x_m
\]
and hence the formula $\exists x_1 \st \Phi(x_1,\ldots,x_m)$ is equivalent to
\[
\bigwedge_{i \in I, j \in J}
 \frac{b_i}{a_{i,1}} - \frac{a_{i,2}}{a_{i,1}} x_2 - \cdots - \frac{a_{i,m}}{a_{i,1}}x_m
\sim_{i,j} 
\frac{b_j}{a_{j,1}} - \frac{a_{j,2}}{a_{j,1}} x_2 - \cdots - \frac{a_{j,m}}{a_{j,1}}x_m
\]
where ${\sim_{i,j}}$ is $\le$ if both ${\sim_i}$ and ${\sim_j}$ are non-strict, and ${\sim_{i,j}}$ is $<$ otherwise.
Hence we can eliminate the variable $x_1$.

Now suppose we are given a Boolean combination $\Psi$ of affine inequalities in real variables $x_1,\ldots,x_m$.
Let $A$ be the set of all constants that appear as the coefficient of some $x_i$ in $\Psi$ (in $\Phi$ above, this is the set of all $a_{i,k}, a_{j,k}$), and $B$ be the set of all other constants (in $\Phi$ above, this is the set of all $b_i,b_j$).
By iteratively applying the process described above, we can construct a Boolean combination $\Gamma$ of linear inequalities (without any free variables) in the elements of $B$ with coefficients that are polynomial in the elements of $A$ such that $\Gamma$ is true if and only if $\exists x_1,\ldots,x_m \st \Psi(x_1,\ldots,x_m)$ is true.

\section{A toolbox for proving MSO decidability}
\label{sec::automata-theory}

We recall classical results in \Cref{sec::unif-rec-words,sec::pup-words}.
Thereafter we give our main new results concerning MSO decidability.

\subsection{Uniformly recurrent words}
\label{sec::unif-rec-words}

A word $\alpha \in \Sigma^\omega$ is
\begin{itemize}
	\item \emph{recurrent} if every factor $u \in \Sigma^*$ of $\alpha$ occurs infinitely often in $\alpha$, and
	\item \emph{uniformly recurrent} if it is recurrent and for every factor $u$, the gaps between consecutive occurrences of $u$ in $\alpha$ are bounded.
	In this case, we define $R_\alpha(n)$ to be the smallest integer $R$ such that every factor $u$ of $\alpha$ with $|u| = n$ appears in $\alpha[k, k+R)$ for all $k$.
	The function $R_\alpha$ is called the \emph{recurrence function} of $\alpha$.
\end{itemize}

Prominent examples of uniformly recurrent words include the Thue-Morse word \cite[Chap.~1]{allouche2003automatic} and all Sturmian words \cite[Chap.~2]{lothaire2002algebraic}.
The following is due to Sem\"enov \cite{semenov84_logic_theor_one_place_funct}, and a modernized version of this proof is given in~\cite[Section 3.1]{karimov2023algorithmic}.\footnote{In fact, Sem\"enov's result works for the more general class of \emph{effectively almost-periodic} words, in which every factor occurs either finitely many times, or with bounded gaps.}

\begin{theorem}
	\label{thm::semenov}
	Suppose we are given an automaton $\Acal$ and a uniformly recurrent word $\alpha$, represented by (i) an oracle to compute $\alpha(n)$ given $n$, and (ii) an oracle to compute $R_\alpha(n)$ given $n$.
	Then we can decide whether $\Acal$ accepts $\alpha$.
\end{theorem}
\begin{corollary}
	\label{thm::semenov-2}
	Suppose we are given an automaton $\Acal$ and a uniformly recurrent word $\alpha$ represented by the oracle (i) above and either 
	\begin{itemize}
		\item[(iii)] an oracle to decide whether a given finite word $u$ occurs in $\alpha$, or
		\item[(iv)] an oracle to compute the factor complexity $\pi_\alpha(n)$ given $n$.
	\end{itemize}
	Then we can decide whether $\Acal$ accepts $\alpha$.
\end{corollary}
\begin{proof}
	Suppose we can decide whether a given $u$ occurs in $\alpha$.
	Then for any $k$, we can compute the set $L(k)$ of all factors of $\alpha$ of length $k$.
	Thus we can compute $R_\alpha(n)$ by computing $L(k)$ for increasing values of $k \ge n$ until we arrive at a set of finite words all of which contain all of $L(n)$ as factors.
	
	Now suppose we can compute $\pi_\alpha(n)$ given $n$.
	Then we can decide whether $u$ occurs in $\alpha$ by first computing $\pi_\alpha(|u|)$, and then computing increasingly larger prefixes of $\alpha$ until we have seen $\pi_\alpha(|u|)$ different factors of length $|u|$.
	At that point we can decide whether $u$ occurs in $\alpha$.
\end{proof}

\subsection{Profinitely ultimately periodic words}
\label{sec::pup-words}

Carton and Thomas \cite{cartonthomas} introduced the class of profinitely ultimately periodic words as a framework to generalise the \emph{contraction method} of Elgot and Rabin \cite{elgotrabin}.
Let $\Sigma$ be an alphabet.

\begin{itemize}
	\item[(a)] A sequence of finite words $\tuple{u_n}_{n \in \naturals}$ is \emph{effectively profinitely ultimately periodic} \cite{cartonthomas} if for any morphism $h\st \Sigma^* \to M$ into a finite monoid $M$, we can compute integers $N, p$ with $p \ge 1$ such that for all $n \ge N$, $h(u_n) = h(u_{n+p})$.
	\item[(b)] An infinite word $\alpha$ is \emph{effectively profinitely ultimately periodic} \cite{rabinovich} if it can be effectively factorised (by, e.g., a Turing machine that is given $\alpha$ as input) as an infinite concatenation $u_0u_1\cdots$ of finite non-empty words forming an effectively profinitely ultimately periodic sequence.
	\item[(c)] We say that a strictly increasing function $f \st \nat \to \nat$ is \emph{effectively contractive} if for any morphism $h \st \nat \to M$ into a finite monoid, the sequence $\seq{h(f(n+1)-f(n)-1)}$ is ultimately periodic with computable period and pre-period.
\end{itemize}

We have the following results about MSO decidability.

\begin{theorem}
	\label{pup-decidable}
	Let $\alpha \in \Sigma^\omega$, and suppose we can compute $\alpha(n)$ given $n$.
	The problem $\Acc{\alpha}$ is decidable if and only if $\alpha$ is effectively profinitely ultimately periodic.
\end{theorem}

The ``if'' direction of \Cref{pup-decidable} is due to Carton and Thomas \cite{cartonthomas}, and the converse is due to Rabinovich \cite{rabinovich}. Every infinite word is profinitely ultimately periodic, as a close inspection of the use of Ramsey's theorem in Rabinovich's proof reveals.
The effectiveness distinguishes words whose automaton acceptance problem is decidable.
A comprehensive class of predicates whose characteristic words ($\alpha \in \{0, 1\}^\omega$) are effectively profinitely ultimately periodic is identified by \cite[Thm.~5.2]{cartonthomas}. This class includes predicates $\{p(n)k^n  \ge 0 \st n \in \nat\}$ for any integer $k \ge 1$ and polynomial $p \in \intg[x]$.

Characteristic words of predicates obtained from effectively contractive and computable $f$ are more restrictive than effectively profinitely ultimately periodic words: in the latter case, we are free to choose the factorisation, whereas in the former case, the factorisation is the unique one into words of the form $0^k1$ for some $k$.
It turns out, however, that effectively contractible words have better compositional properties (\Cref{sec::compositions}).
For now, we record the following corollary of \Cref{pup-decidable}.

\begin{corollary}\label{cor::epup-functions}
	The MSO theory of $\langle \nat; <, P\rangle$, where $P = \{f(n) \st n\in\nat\}$ for a computable and effectively contractive function $f$, is decidable.
\end{corollary}
\begin{proof}
	The characteristic word $\alpha$ of $P$ is effectively profinitely ultimately periodic, with the required factorisation of $\alpha$ being $\langle 0^{f(0)} 1, 0^{f(1)-f(0)-1} 1, 0^{f(2)-f(1)-1} 1 , \ldots \rangle$.
\end{proof}

\subsection{Disjunctive words}\label{Sec: Normal Words}

Recall that a word $\alpha \in \Sigma^\omega$ is \emph{disjunctive} if every $u \in~\Sigma^*$ appears infinitely often in $\alpha$. %
Disjunctivity is usually considered when $\Sigma = \{0,\dots,b-1\}$ is the alphabet of digits and $\alpha$ is the base-$b$ expansion of a real number $x \in [0,1)$.
Thus, when $x = \sqrt{2}-1 = 0.41421356\cdots$ and $b=10$, $\alpha = 41421356 \cdots$.
We have the following result about such words.
\begin{theorem}\label{normal-decidable}
	If $\alpha$ is disjunctive and $\alpha(n)$ can be effectively computed given $n$, then $\Acc{\alpha}$ is decidable.
\end{theorem}
Intuitively, the proof uses the abundance of all factors in $\alpha$ to deduce that the set of states visited infinitely often is an entire bottom strongly connected component in the graph induced by the automaton.
\begin{proof}
	Consider an automaton $\automaton$ as a directed graph allowing multiple edges. We partition the graph into its strongly connected components (SCCs) and call an SCC without outgoing edges a bottom SCC\@. We will show that the set of states visited infinitely often by the run of $\automaton$ on $\alpha$ is precisely a bottom SCC\@. Hence we can decide $\Acc{\alpha}$ by simulating the run of given $\Acal$ on $\alpha$ until a bottom SCC is reached. Then $\alpha$ is accepted if and only if this bottom SCC is in the (Muller) acceptance condition of $\Acal$.
	
	We need to show that (a) if an SCC is not a bottom SCC, then the run eventually exits it; and (b) if the run enters a bottom SCC, all of its states are visited infinitely often.
	To prove~(a), let $S$ be a non-bottom SCC\@. \
	Enumerate the states of~$S$ as $q_1,\ldots,q_n$.
	Inductively construct a sequence of words $u_1, \ldots, u_n \in \Sigma^+$, starting with $k = 1$, such that for all $1 \le k \le n$, $\delta(q_k, u_1 \cdots u_k) \notin S$.
	Here, $u_{k+1}$ can be chosen to be any word such that takes the automaton from the state $\delta(q_{k+1}, u_1\cdots u_k)$ to state outside $S$.
	Then the word $u = u_1\cdots u_n$ is such that for any $q \in S$, $\delta(q,u) \notin S$.
	By disjunctivity, $u$ occurs in $\alpha$.
	It follows that when reading $\alpha$ from any $q \in S$ the run will exit $S$.
	
	We prove (b) similarly. 
	Let $S$ be a bottom SCC consisting of the states $q_1,\ldots,q_n$, and $q \in S$.
	Inductively construct $u_1, \ldots, u_n \in \Sigma^+$ such that $\delta(q_k, u_1\cdots u_k) = q$ for all $k$.
	Here, $u_{k+1}$ can be chosen to be any word that takes the automaton from the state $\delta(q_{k+1}, u_1 \cdots u_{k})$ to the state $q$.
	It follows that whenever the word $u = u_1\cdots u_n$ is read, regardless of from which state in $S$, the state $q$ will be visited.
	By disjunctivity we have that every $q \in S$ is visited infinitely often.
\end{proof}

\subsection{Closure under transductions}
\label{sec::closure}

In this section, when we say a word $\alpha \in \Sigma^\omega$ is given, we mean that an oracle is given that computes $\alpha(n)$ on input~$n$.
Our first result is that automaton acceptance for an infinite word $\alpha = \Tcal(\beta)$, where $\Tcal$ is a transducer, can be reduced to automaton acceptance for $\beta$. 

\begin{lemma}\label{thm::transduction}
	Suppose we are given $\alpha \in \Sigma_1^\omega$, $\beta \in \Sigma_2^\omega$, an automaton $\Acal$ over $\Sigma_1$, and a transducer $\Tcal$ such that $\alpha = \Tcal(\beta)$.
	We can compute an automaton $\Bcal$ such that $\alpha \in L(\Acal) \Leftrightarrow \beta \in L(\Bcal)$.
\end{lemma}
\begin{proof}
	Write $\automaton = (Q, \qinit, \delta, \Fcal)$ and $\Tcal = (R, \rinit, \sigma)$.
	The automaton $\Bcal$ simulates what $\Tcal$ would do upon reading~$\alpha$, and furthermore, what $\Acal$ would do upon reading each output block of $\Tcal(\alpha)$.

	The set of states of $\Bcal$ is $Q' = J \times R$, where $J$ is the set of journeys in $\Acal$.
	Each state keeps track of the last journey made in $\Acal$ and the current state in~$\Tcal$.
	We further define $\qinit' = ((\qinit, \qinit, \emptyset), \rinit)$ and
	\[
	\delta'(((p, q, V), r), b) = (\journey(q, u), r')
	\]
	where $\sigma(r, b) = (r', u)$.
	By \Cref{journeys}, a state $s$ appears infinitely often in $\Acal(\alpha)$ if and only if a state $((p,q,V), r)$ with $s \in V$ appears infinitely often in $\Bcal(\beta)$.
	We therefore define $\Fcal'$ by
	\[
	F' \in \mathcal{F}' \iff \bigg(\bigcup_{((p, q, V), r) \in F'} V \: \bigg) \in \mathcal{F} 
	\]
	and $\Bcal = (Q', \qinit', \delta', \Fcal')$.
\end{proof}
\begin{corollary}\label{cor::transducer reduction}
	Let $\alpha \in \Sigma_1^\omega$, $\beta \in \Sigma_2^\omega$, and suppose that $\alpha = \Tcal(\beta)$ for a transducer $\Tcal$.
	Then $\Acc{\alpha}$ reduces to $\Acc{\beta}$.
\end{corollary}

We next give a generalisation of \Cref{thm::transduction} that we will be crucial throughout the paper.
Here we transduce from $\beta$ not into $\alpha$ but rather to ``compressed'' versions thereof, where the compression is performed by a morphism into a finite monoid. 

\begin{theorem}\label{thm::fancy-transduction-taylor's-version}
	Suppose we are given $\alpha \in \Sigma_1^\omega$, $\beta \in \Sigma_2^\omega$, an automaton $\Acal$ over $\Sigma_1$, and an oracle that, given a morphism $h \st \Sigma_1^* \to M$ into a finite monoid $M$, computes a transducer $\Tcal$ (with input alphabet $\Sigma_2$ and output alphabet $M$) with the following property:
	there exists a factorisation $\alpha = \prod_{n=0}^\infty w_n$ such that $\Tcal(\beta) = h(w_0) h(w_1) h(w_2) \cdots$.
	Then we can compute an automaton $\Bcal$ such that $\alpha \in L(\Acal) \Leftrightarrow \beta \in L(\Bcal)$.
\end{theorem}
\begin{proof}
	Let $J$ be the set of all journeys in $\Acal$, $M$ be the journey monoid of $\Acal$, $h \st \Sigma_1^* \to M$ be the morphism mapping each finite word to the equivalence class of its set of journeys in $\Acal$, $\Tcal$ be the corresponding transducer (given in the statement of \Cref{thm::fancy-transduction-taylor's-version}), and $\alpha = \prod_{n=0}^\infty w_n$ be the corresponding factorisation.
	Write $\Acal = (Q, \qinit, \delta, \Fcal)$ and $\Tcal = (R, \rinit, \sigma)$.
	We construct an automaton $\Mcal = (J, (\qinit, \qinit, \varnothing), \delta', \Fcal')$ over the alphabet $M$ as follows.
	The states simply record the last journey taken.
	For $(p,q,V) \in J$ and $x \in M$, we define $\delta'((p,q,V), x) = \journey(q, x)$.
	Finally, we define the acceptance condition by
	\[
	F' \in \mathcal{F}' \iff \bigg(\bigcup_{(q, r, V) \in F'} V \: \bigg) \in \mathcal{F}.
	\]
	Then $\Mcal(\Tcal(\beta))$ is factorisation of the run of $\Acal$ on $\alpha$ into journeys as in \Cref{journeys}, and $\alpha \in L(\Acal)  \Leftrightarrow \Tcal(\beta) \in L(\Mcal)$.
	Applying \Cref{thm::transduction} with $\beta$, $\Tcal(\beta)$ and $\Mcal$, we construct an automaton $\Bcal$ such that $\Tcal(\beta) \in L(\Mcal) \Leftrightarrow \beta \in L(\Bcal)$.
\end{proof}

\subsection{Sparse procyclic predicates}

We now discuss the main automata-theoretic tool that will be used \Cref{sec::lrs-with-one-dominant-root}.
The idea is that, when considering the automaton acceptance problem, we reduce from the characteristic word to its compressed version, called the \emph{order word}.
Let $P_1,\ldots,P_d$ be infinite predicates with the characteristic word $\alpha \in (\{0,1\}^d)^\omega$.
Order the elements of $P_i$ as $\seq{p^{(i)}_n}$.

\begin{definition}[Order Word]
	The order word of $P_1,\ldots,P_d$, written $\order{P_1, \dots, P_d}$, is the infinite word obtained by deleting all occurrences of $\zerovec$ from $\alpha$.
\end{definition}

In compressing the characteristic word $\alpha$ to the order word $\beta$, one only retains partial information, i.e.\ a particular aspect of the interaction between predicates. 
Not surprisingly, by an immediate application of \Cref{thm::transduction},  given an automaton $\Bcal$ we can construct an automaton~$\Acal$ such that $\alpha \in L(\Acal) \Leftrightarrow \beta \in L(\Bcal)$ for any characteristic word $\alpha$.
In particular, $\Acc{\beta}$ reduces to $\Acc{\alpha}$.
Remarkably however, under certain circumstances, the order word captures all the information we need to decide automaton acceptance, and we can perform the  reverse reduction.
To see this, write
\begin{equation*}
	\alpha = \mathbf{0}^{k_0}\beta(0)\cdots\mathbf{0}^{k_n}\beta(n)\cdots
\end{equation*}
and suppose we want to decide whether $\alpha \in L(\Acal)$.
The automaton $\Acal$ is finite, and consequently one can compute $L, p > 0$ such that for all $n \ge L$, it cannot distinguish $\zerovec^{n}$ from~$\zerovec^{n+p}$.
Provided that $k_n$ is persistently larger than $L$, it suffices to only keep track of $k_n$ modulo $p$.
That is, we can construct an automaton $\Acal'$ such that $\alpha \in L(\Acal) \Leftrightarrow \alpha' \in L(\Acal')$ where $\alpha' = \zerovec^{m_0} \beta(0) \cdots \zerovec^{m_n}\beta(n)\cdots$ and each $m_n$ is small but indistinguishable from $k_n$ by $\Acal$.
We can then hope to insert the sequence $\seq{\zerovec^{m_n}}$ into $\beta$ using a transducer $\Ccal$, i.e., $\alpha' = \Ccal(\beta)$.
Then by \Cref{thm::transduction}, deciding whether $\alpha' \in L(\Acal')$ reduces to deciding whether $\beta \in L(\Bcal)$ for an automaton $\Bcal$ computed from $\Acal'$ and $\Ccal$.
We next formalise this argument.
We say that a sequence $\seq{u_n}$ of integers is \emph{effectively procyclic} if, given $m \ge 1$, we can compute $N, p$ such that $u_{n+p} \equiv u_n \, (\bmod \, m)$ for all $n \ge N$.

\begin{definition}[Effectively Procyclic Predicates]
	\label{procyclicpredicate}
	The predicates $P_1, \ldots, P_d$ are effectively procyclic if each $\seq{p^{(i)}_n}$ is effectively procyclic.
\end{definition}

We briefly argue that for a single predicate $P$, the notions of being effectively procyclic and having an effectively profinitely ultimately periodic characteristic word are incomparable.
Firstly, let $\alpha \in \{0,1\}^\omega$ be such that $\alpha(n)$ can be computed given $n$, but $\Acc{\alpha}$ is undecidable.
For example, we can take $\alpha = \prod_{i=0}^\infty v_{\sigma(i)}$ where $\sigma \st \nat \to (\nat_{>0})^2$ enumerates all possible pairs $(t, k)$ consisting of a Turing machine $t$ (itself represented by a positive integer; we assume the initial and halting states are always distinct) and a positive integer $k$, and
\[
v_{\sigma(i)} = 
\begin{cases}
	0^k 1^t \textrm{ if $t$ halts on $\varepsilon$ exactly after $k$ steps,}\\
	\varepsilon, \textrm{ otherwise.}
\end{cases}
\]
Then the Turing machine $t \in \nat_{>0}$ halts on empty input if and only if $10^t1$ occurs in $\alpha$.
Note that in this example, in fact, the first-order theory of $\langle \nat; <, \alpha\rangle$ is undecidable.

By undecidability of $\Acc{\alpha}$, both $0$ and $1$ must occur infinitely often in $\alpha$.
Construct $\beta$ by replacing in $\alpha$ each 0 by $10$ and the $k$th occurrence of $1$ by $1 0^{k! + 1}$.
Let $P = \{n \st \beta(n) = 1\}$, which is effectively procyclic by construction.
However, since $\alpha = \Bcal(\beta)$ for a transducer, by \Cref{thm::transduction}, $\Acc{\beta}$ and hence the MSO theory of $\langle \nat; <, P \rangle$ must be undecidable.
Hence $\beta$ cannot be effectively profinitely ultimately periodic by \Cref{pup-decidable}.

Now suppose $\alpha \in \{0,1\}^\omega$ is a non-periodic word such that $\Acc{\alpha}$ is decidable, e.g., the Thue-Morse word.
Let $\beta = 0^{\alpha(0)} 1 0^{\alpha(1)}1 0^{\alpha(2)} 1\cdots$, $P = \{n \in \nat \st \beta(n) = 1\}$, and $\seq{p_n}$ be the ordering of $P$.
Observe that $p_{n+1} - p_n \in \{1,2\}$ for all $n$.
By construction, for all $m \ge 2$, $\seq{p_n \bmod m}$ is not ultimately periodic, which implies that $P$ is not procyclic.
However, by a transduction argument, $\Acc{\alpha}$ is Turing-equivalent to $\Acc{\beta}$.
Therefore, $\Acc{\beta}$ is decidable and hence $\beta$ is effectively profinitely ultimately periodic.

\begin{definition}[Effectively Sparse Predicates]
	\label{sparsepredicate}
	The predicates $P_1,\ldots,P_d$ are effectively sparse if for every $i,j$ and $K$, $0 < |p^{(i)}_n - p^{(j)}_m| < K$ has only finitely many solutions that can be effectively determined.
\end{definition}

This is equivalent to effective divergence of $k_n$ to $+\infty$: for every $K$, we can compute $N$ such that for all $n \ge N$, $k_n \ge K$.
We are now ready to state our reduction from characteristic words to order words.
We assume that effectively procyclic and effectively sparse predicates $P_1,\ldots,P_d$ are given by oracles to check whether $n \in P_i$ for all $n \in \nat$ and $1 \le i \le d$, as well as the oracles described in the definitions above.

\begin{theorem}
	\label{thm::sparse-procyclic-reduction-to-order-word}
	There exists an algorithm that takes an automaton $\Acal$ and effectively procyclic and effectively sparse predicates $P_1,\ldots,P_d$, and outputs an automaton $\Bcal$ such that $\alpha \in L(\Acal) \Leftrightarrow \beta \in L(\Bcal)$, where $\alpha$ is the characteristic word of $P_1,\ldots,P_d$ and $\beta$ is the corresponding order word.
\end{theorem}
\begin{proof}
	We factorise $\alpha = \prod_{n=0}^\infty w_n$, where $w_n = \zerovec^{k_n} \beta(n)$.
	By \Cref{thm::fancy-transduction-taylor's-version} it suffices to show how to construct, given a morphism $h \st (\{0,1\}^d)^* \to M$, a transducer $\Tcal$ such that $\Tcal(\beta) = h(w_0) h(w_1) h(w_2) \cdots$.
	By the classical lasso argument, from $M$ we can construct $L, m \ge 1$ such that for any $x \in M$ and $k \ge L + m$, $x^k = x^{L + (k-L) \bmod m}$.
	From effective sparsity, we can compute $N$ such that for all $n \ge N$, $k_n > L$.
	
	Define, for all $n \ge 0$ and $1 \le i \le d$,
	\begin{itemize}
		\item $t_n = \sum_{i=0}^{n} (k_i + 1)$, which is equal to $|w_0| + \cdots + |w_n|$,
		\item $l^{(i)}_n =  \max \{t_k \le t_n \st \beta(k) = (b_1,\ldots,b_d), b_i = 1\} - 1$, i.e.\ the largest $j \le t_n$ such that $j \in P_i$, and
		\item $s^{(i)}_n =  \min \{t_k > t_n \st \beta(k) = (b_1,\ldots,b_d), b_i = 1\} - 1$, i.e.\ the smallest $j > t_n$ such that $j \in P_i$.
	\end{itemize}
	Since $P_1,\ldots,P_d$ are effectively procyclic, $(s^{(i)}_n - l^{(i)}_n)$ is periodic modulo every $m$ with a computable period and pre-period.
	
	The transducer $\Tcal$ has the values of $h(w_0), \ldots, h(w_N)$ hard-coded into it.
	Before reading $\beta(n+1)$ for $n \ge N$, it has in memory $t_n \bmod m$ for each $1 \le i \le d$, as well as the values of $l^{(i)}_n \bmod m$ and $\big( s^{(i)}_n - l^{(i)}_n \big) \bmod m$, the latter via effective procyclicity.
	Upon reading $\beta(n+1) = (b_1,\ldots,b_d)$, it first determines some $i$ such that $b_i = 1$.
	At this point we have that $s^{(i)}_n = t_{n+1} - 1$
	and hence $t_{n+1} = l^{(i)}_n + (s^{(i)}_n - l^{(i)}_n) + 1$.
    Therefore, the transducer can compute  $t_{n+1} \bmod m$ using the values it has in memory.
    Next, it computes $k_{n+1} \bmod m$ using the relation $k_{n+1} = t_{n+1}-t_n-1$.
	Since $k_{n+1} \ge L + m$, we have that
	\[
	h(w_n) = (h(\zerovec))^{L + (k_{n+1} - L) \bmod m} \cdot h(\beta(n))
	\]
	which can now be computed and output by $\Tcal$.
	Finally, the transducer computes the values of $l^{(i)}_{n+1} \bmod m$ as well as $s^{(i)}_{n+1} - l^{(i)}_{n+1} \bmod m$ for all $i$, by checking whether $b_i = 1$ or $b_i = 0$ and using the value of $k_{n+1} \bmod m$.
\end{proof}

\begin{corollary}\label{cor::sparse-char-eq-order}
	Let $P_1,\ldots,P_d$ be effectively procyclic and effectively sparse predicates with the characteristic word $\alpha$ and the order word $\beta$.
	The problems $\Acc{\alpha}$ and $\Acc{\beta}$ are Turing-equivalent.
\end{corollary}
\begin{proof}
	That $\Acc{\alpha}$ reduces to $\Acc{\beta}$ follows from \Cref{thm::sparse-procyclic-reduction-to-order-word}, and that $\Acc{\beta}$ reduces to $\Acc{\alpha}$ is true for any predicates $P_1,\ldots,P_d$ as discussed earlier.
\end{proof}

\subsection{Closure under compositions}
\label{sec::compositions}

We now study predicates linked by function composition.
\begin{theorem}\label{thm::chains}
	Let $g_1,\ldots,g_d\st\nat\to\nat$ be effectively contractive functions and $f_i = g_1 \circ \cdots \circ g_i$ for $1 \le i \le d$.
	The MSO theory of $\langle \nat; <, P_1,\ldots,P_d\rangle$, where $P_i = \{f_i(n) \st n \in \nat\}$, is decidable.
\end{theorem}
\begin{proof}
	Observe that $f_1(n) = g_1(n)$, $f_2(n) = g_1(g_2(n))$, $f_3(n) = g_1(g_2(g_3(n)))$, and so on.
	Hence $P_d \subseteq \cdots \subseteq P_1$.
	
	We proceed by induction.
	If $d = 1$, then decidability follows from \Cref{cor::epup-functions}.
	Next, consider $g_1,\ldots,g_d$ for some $d > 1$.
	Let $\Sigma = (\{0,1\})^{d}$,  $\alpha \in \Sigma^\omega$ be the characteristic word of $\langle \nat; <, P_1,\ldots,P_d\rangle$, and $\beta \in \Sigma^\omega$ be the order word.
	Further let $h_i = g_2 \circ \cdots \circ g_i$ for $2 \le i \le d$, $Q_i = \{h_i(n) \st n \in \nat\}$, and $\sigma$ be the characteristic word of $\langle \nat; <, Q_2,\ldots,Q_d\rangle$.
	By the induction hypothesis, $\Acc{\sigma}$ is decidable.
	We will show how to transduce from $\sigma$ into (any compressed version, as appropriate, of) $\alpha$.
	The decidability then follows from \Cref{thm::fancy-transduction-taylor's-version}.
	
	Let $h \st \Sigma^* \to M$ be a morphism into a finite monoid $M$.
	For $n \in \nat$, write $\beta(n) = (b_{n,1}, \ldots, b_{n,d})$ and $\sigma_n = (s_{n,2},\ldots,s_{n,d})$.
	Observe that $b_{n,k} = s_{n,k}$ for all $2 \le k\le d$.
	We factorise $\alpha = \prod_{n=0}^\infty w_n$ where $w_n = \zerovec^{k_n} \beta(n)$, $k_0 = g_1(0)$, and $k_{n} = g(n) - g(n-1) -1$ for all $n \ge 1$.
	By the assumption on $g_1 = f_1$, $\seq{k_n}$ is effectively profinitely ultimately periodic.
	In particular, a transducer can compute $\seq{h(\zerovec^{k_n})}$ at the $n$th step.
	Therefore, we can construct a transducer $\Tcal$ that, upon reading $\sigma(n)$, outputs $h(\zerovec^{k_n}) h((1,s_{n,2},\ldots,s_{n,d}))$.
	Then $\Tcal(\sigma) = h(w_0)h(w_1)\cdots$.
\end{proof}

We will not use \Cref{thm::chains} to prove any of our main theorems.
However, we believe that it is of independent interest, as it gives us other new decidability results: for example, if we take $g_1(n) = n_2$, $g_2(n) = 2^n$, and $g_3(n) = 3^n$, we obtain that the MSO theory of $\langle \nat; <, \kthpowers{2}, \powersofk{4}, \{2^{2\cdot 3^n} \st n \in \nat\}\rangle$ is decidable.
By a similar argument, we obtain that the MSO theory of 
$\langle \nat; <, \{(k^d)^n \st n\in\nat\}, \kthpowers{d}\rangle$ is decidable for any integers $k,d \ge 2$, which also follows from our second main result (see \Cref{thm::main-2-cor-3}).

\section{Linear recurrence sequences with a single non-repeated dominant root}

\label{sec::lrs-with-one-dominant-root}

\subsection{Overview}
\label{sec::4-1}

By the\emph{value set} of an integer sequence $\seq{u_n}$ we mean the set $\{u_n \ge 0 \st n \in \nat\}$.
In this section, we consider integer LRS $\seq{u^{(i)}_n}$ with respective value sets $P_i \subseteq \nat$ for $1 \le i \le d$ that satisfy the following condition.

\begin{itemize}
	\item[(A)] Each $\seq{u^{(i)}_n}$ has the exponential polynomial representation
	\[
	u^{(i)}_n = c_i \rho_i^n + \sum_{j=1}^{k_i} p_{i,j}(n)\rho_{i,j}^n
	\]
	and the dominant part $\seq{c_i \rho_i^n}$, where $c_i > 0$ and $\rho_i > 1$.
	That is, $\seq{u^{(i)}_n}$ is non-constant, has the single non-repeated dominant root $\rho_i$, and $P_i$ is infinite.
\end{itemize}
Note that if $\rho_i = 1$, because $\seq{u^{(i)}_n}$ takes integer values, $P_i$ must be finite.
The same conclusion holds if $c_i < 0$.
Such~$P_i$ can be defined in the structure $\langle \nat; < \rangle$, and are not of interest to us.
Our main decision problem is the following.

\begin{problem}
	\label{problem-mso-1-dom-root}
	Given $\seq{u^{(1)}_n},\ldots, \seq{u^{(d)}_n}$ as above and an MSO formula $\varphi$, decide whether $\langle \nat; <, P_1, \ldots, P_d \rangle \models \varphi$.
\end{problem}

When $d=1$, decidability of Problem~\ref{problem-mso-1-dom-root} can be shown using the contraction method of Elgot and Rabin.
(The decidability also follows from the main results in this section.)
On the other hand, it is easily shown that Problem~\ref{problem-mso-1-dom-root}, even in case of $d=2$, subsumes many major open decision problems about integer LRS.

\begin{lemma}
	\label{thm::sign-pattern-to-prob-1}
	Let $\Acal$ be an automaton over $\Sigma \coloneqq \{-,0,+\}$ and $\seq{v_n}$ an integer LRS with sign pattern $\sigma \in \Sigma^\omega$, where $\sigma(n)$ is defined by the sign of $v_n$.
	We can construct an automaton $\Bcal$ and integer LRS $\seq{u^{(1)}_n}, \seq{u^{(2)}_n}$ satisfying condition~(A) with respective value sets $P_1,P_2$ such that $\sigma \in L(\Acal) \Leftrightarrow \tau \in L(\Bcal)$, where $\tau$ is the characteristic word of $P_1,P_2$.
\end{lemma}
\begin{proof}
	Let $\rho \in \nat_{>1}$ be such that $|v_n| = o(\rho^n)$ (see \Cref{lem:growth LRS}). 
	Construct $\varepsilon \in (0,1)$ such that $\frac 1 \rho (1+\varepsilon) < (1-\varepsilon) < (1+\varepsilon) < \rho$, and $C \in \nat_{>0}$ such that $|v_n| < \varepsilon \cdot C\rho^n$ for all $n$. 
	Define $u^{(1)}_n = C \rho^n$ and $u^{(2)}_n = C\rho^n + v_n$.
	We give a transducer~$\Gamma$ such that $\Gamma(\tau) = \sigma$, and then invoke \Cref{thm::transduction}.

	By construction of $C, \rho$ and $\varepsilon$, for all $n \ge 0$, $u^{(2)}_n$ is the unique term in the sequence $\seq{u^{(2)}_n}$ belonging to the interval $((1-\varepsilon) \cdot u^{(1)}_n, (1+\varepsilon) \cdot u^{(1)}_n)$.
	Moreover, by construction of $\varepsilon$ these intervals are disjoint and ordered according to $n$.
	Hence the order word $\beta$ of $P_1,P_2$ can be uniquely factorised as $\beta = w_0 w_1 w_2 \cdots$ where each $w_n$ is either $(1,1)$ (when $v_n = 0$), or $(0,1)(1,0)$ (when $v_n > 0$), or $(1,0)(0,1)$ (when $v_n < 0$).
	The transducer $\Gamma$ simply discards all occurrences of $\zerovec$ from $\tau$, performs the factorisation above, and outputs the corresponding sign $-$, $0$ or $+$ for each factor $w_n$.
\end{proof}

By the MSO theory of the sign pattern of an LRS $\seq{v_n}$ we mean the MSO theory of the structure $\langle \nat; <, P_-, P_0, P_+\rangle$ where $P_-, P_0, P_+$ partition $\nat$ and the unique predicate $P \in \{P_-, P_0, P_+\}$ containing $n$ is determined by the sign of $v_n$ (i.e., $n \in P_-$ if and only if $v_n < 0$).
By querying this structure using MSO formulas, we can ask whether $v_n = 0$ for some $n$, whether $v_n \ge 0$ for all (sufficiently large) $n$, etc.
The two corollaries of \Cref{thm::sign-pattern-to-prob-1} below follow immediately from B\"uchi's theorem (\Cref{thm:MSOautomaton}) and the definitions of Skolem and other problems of LRS (\Cref{sec::lrs}).

\begin{corollary}
	Given an integer LRS $\seq{v_n}$ and an MSO formula $\varphi$, we can construct integer LRS $\seq{u^{(1)}_n}, \seq{u^{(2)}_n}$ and an MSO formula $\psi$ such that $\langle \nat; <, P_-, P_0, P_+\rangle \models \varphi \Leftrightarrow \langle \nat; <, P_1, P_2 \rangle \models \psi$.
\end{corollary}

\begin{corollary}
	\label{thm::pos-hardness}
	The Skolem Problem, the Positivity Problem, and the Ultimate Positivity Problem reduce to Problem 1.
\end{corollary}
These results establish that Problem 1 is hard because is subsumes the problem of model checking sign patterns of arbitrary integer LRS.
It is natural to ask: in the context of sign patterns, how hard  is full model checking in comparison to just, for example, the Skolem Problem?
Surprisingly, the problem of deciding, given an MSO formula $\varphi$ and a \emph{diagonalisable} integer LRS $\seq{v_n}$,  whether $\langle \nat; <, P_-, P_0, P_+\rangle \models \varphi$ is Turing-equivalent to the Positivity Problem for \emph{diagonalisable} integer LRS; see \cite[Thm.~12]{karimov2023power} and \cite[Lem.~7.0.2]{karimov2023algorithmic}.
For non-diagonalisable LRS, no such result is known.

Examining the construction of \Cref{thm::sign-pattern-to-prob-1}, it is not difficult to see that all three problems, in fact, reduce to the restriction of Problem~\ref{problem-mso-1-dom-root} to first-order formulas $\varphi$.
By a similar argument, the problem of deciding, given two integer LRS $\seq{v_n}$ and $\seq{w_n}$, whether $v_n = w_m$ for some $n,m$ also reduces to Problem 1 with first-order $\varphi$ and $d = 2$.
The latter is another problem about LRS that is wide-open at the moment.

Given the hardness of Problem~\ref{problem-mso-1-dom-root}, we will further restrict the linear recurrence sequences that we consider.
We will be interested in the following conditions in addition to (A).
\begin{itemize}
	\item[(B)] For any $i \ne j$, $c_i\rho_i^n = c_j \rho_j^m$ holds for finitely many $(n,m) \in \nat^2$.
	Equivalently, for every $i \ne j$, the equation $c_i\rho_i^n = c_j\rho_j^m$ has at most one solution $(n,m)$ in~$\nat^2$.
	\item[(C)] The numbers $1/\log(\rho_1),\ldots,1/\log(\rho_d)$ are linearly independent over $\rat$.
\end{itemize}
Observe that (C) implies that $\rho_i, \rho_j$ are multiplicatively independent for every $i \ne j$, which implies (B).
Our main result is the following.

\begin{theorem}
	\label{powers-main-structure-parametrised}
	Problem 1 is decidable for $\seq{u^{(1)}_n},\ldots, \seq{u^{(d)}_n}$ that satisfy the conditions (A) and (C).
\end{theorem}

To see how (C) can be applied in practice, recall from \Cref{lem::rank-d-2-to-LI} that if $\operatorname{rank}(G_M(\rho_1,\ldots,\rho_d)) \ge d-2$ (in particular, if $d\le 2$) and $\rho_1,\ldots,\rho_d$ are pairwise multiplicatively independent, then (C) is satisfied.
Our next main result is that if we assume Schanuel's conjecture, then we can replace the condition (C) with the weaker condition (B).

\begin{theorem}
	\label{thm::main-schanuel}
	Assuming Schanuel's conjecture, Problem~\ref{problem-mso-1-dom-root} is decidable for integer LRS satisfying conditions (A) and (B).
	Moreover, the decision procedure relies on Schanuel's conjecture only for termination.
\end{theorem}

The second part of the statement of \Cref{thm::main-schanuel} means that termination of the decision procedure is guaranteed by Schanuel's conjecture, and whenever the procedure terminates, its output is unconditionally guaranteed to be correct.
As a corollary of \Cref{thm::main-schanuel}, we obtain the following theorem from the Introduction.

\begin{theorem}
	\label{thm::main-schanuel-weak}
	Suppose we are given an MSO formula $\varphi$ and integers $a_1,\rho_1,\ldots,a_d,\rho_d \ge 1$.
	Assuming Schanuel's conjecture, it is decidable whether $\varphi$ holds in $\langle \nat; <, P_1,\ldots,P_d \rangle$, where $P_i = \{a_i\rho_i^n \st n\in\nat\}$.
\end{theorem}

\begin{proof}
	The idea is that, for a collection of LRS of the form $u_n = a \rho^n$, we  get the condition (B) for free by taking subsequences.
	Observe that if $a_i \rho_i^n = a_j \rho_j^m$ has more that one solution in $n,m$, then $\rho_i^b = \rho_j^c$ for some integers $b,c > 0$.
	Therefore, from $\rho_1,\ldots,\rho_d$ we can compute $k \ge 1$ such that, writing $b_i = a_i \rho_i^r$, $\mu_i = \rho_i^k$, and
	\[
	Q_{i,r} = \{a_i \rho_i^{kn+r} \st n \in \nat\} = \{b_i \mu_i^n \st n \in \nat\}
	\] for $1 \le i \le d$ and $0 \le r < k$, we have the following.
	For all $i, i', r, r'$, either $Q_{i,r} = Q_{i',r'}$ or $Q_{i,r} \cap Q_{i',r'}$ is finite. 
	Let $\widetilde{Q}_1, \ldots, \widetilde{Q}_m$ be all the distinct predicates from $\{Q_{i,r} \st 1 \le i \le d, \, 0 \le r < k\}$.
	Then $\widetilde{Q}_i \cap \widetilde{Q}_j$ is finite for all $i \ne j$.
	Since $P_i = Q_{i,0} \cup \cdots \cup Q_{i,k-1}$, we can construct an MSO formula $\psi$ that holds in $\langle \nat; <, \widetilde{Q}_1, \ldots, \widetilde{Q}_m\rangle$ if and only if $\varphi$ holds in $\langle \nat; <, P_1,\ldots,P_d \rangle$.
	It remains to invoke \Cref{thm::main-schanuel}.
\end{proof}

What about decidability of the individual MSO theories $\langle \nat; <, P_1,\ldots,P_d\rangle$?
That is, what happens if we fix the sequences $\seq{u^{(1)}_n},\ldots, \seq{u^{(d)}_n}$, and allow only the MSO formula $\varphi$ to be given as input?
What information about the specific LRS do we need to decide the MSO theory?
The answer is given in \Cref{thm::main-ideal-version-consequence}, which is based on the following theorem.
In short, in this setting too we can replace (C) with the weaker condition (B).

\begin{theorem}
	\label{thm::main-ideal-version}
	Given integer LRS $\seq{u^{(1)}_n},\ldots, \seq{u^{(d)}_n}$ satisfying conditions (A) and (B), the ideal of all polynomial relations between $\log(c_1),\log(\rho_1),\ldots,\log(c_d),\log(\rho_d)$, a basis of $G_L(1/\log(\rho_1),\ldots,1/\log(\rho_d))$, and an MSO formula $\varphi$, we can decide whether $\langle \nat; <, P_1,\ldots,P_d\rangle \models \varphi$.
\end{theorem}
We note that a basis for $G_L(1/\log(\rho_1),\ldots,1/\log(\rho_d))$ can be extracted from the aforementioned ideal by using, for example, Gr\"obner bases and modifying the algorithm for deciding whether a given polynomial occurs in an ideal \cite{Cox2025}: observe that for all $k_1,\ldots,k_d \in \intg$, $\sum_{i=1}^d \frac{k_i}{\log(\rho_i)} = 0$ if and only if $\sum_{i=1}^d k_i \prod_{j \ne i} z_j = 0$.
In our formulation, we chose emphasise that there are two distinct sources of non-uniformity.
\begin{corollary}
	\label{thm::main-ideal-version-consequence}
	Let $\seq{u^{(1)}_n},\ldots, \seq{u^{(d)}_n}$ be integer LRS satisfying conditions (A) and (B).
	Then the MSO theory of $\langle \nat; <, P_1,\ldots,P_d\rangle$ is decidable.
\end{corollary}
\begin{proof}
	The Turing machine for deciding the MSO theory of $\langle \nat; <, P_1,\ldots,P_d\rangle$ simply has the ideal of all polynomial relations between $\log(c_1),\log(\rho_1),\ldots,\log(c_d),\log(\rho_d)$, which has a finite representation thanks to Hilbert basis theorem, hard-coded into it.
	This representation can be effectively operated on using, for example, Gr\"obner bases \cite{Cox2025}.
\end{proof}

From \Cref{thm::main-ideal-version} we also obtain the following result stated in the Introduction.
\begin{theorem}
	\label{thm::main-quirky}
	For any integers $a_1,\rho_1,\ldots,a_d,\rho_d \ge 1$, there exists an algorithm that, given an MSO formula $\varphi$, decides whether $\varphi$ holds in $\langle \nat; <, P_1,\ldots,P_d \rangle$, where $P_i = \{a_i\rho_i^n \st n\in\nat\}$.
\end{theorem}

\begin{proof}
	Construct $k \ge 1$ and $\widetilde{Q}_1, \ldots, \widetilde{Q}_m$ as in the proof of \Cref{thm::main-schanuel-weak}.
	By \Cref{thm::main-ideal-version-consequence}, the MSO theory of $\langle \nat; <, \widetilde{Q}_1, \ldots, \widetilde{Q}_m\rangle$ is decidable.
	Given $\varphi$, construct $\psi$ as in the proof of \Cref{thm::main-schanuel-weak}.
	We have that $\varphi$ holds in $\langle\nat; <, P_1,\ldots,P_d\rangle$ if and only if $\psi$ holds in $\langle \nat; <, \widetilde{Q}_1, \ldots, \widetilde{Q}_m\rangle$.
	It remains to check whether the latter is the case.
\end{proof}

We mention that the first part of the statement of \Cref{thm::main-schanuel} immediately follows from \Cref{thm::main-ideal-version}: Schanuel's conjecture tells us that all polynomial relations between $\log(c_1),\log(\rho_1),\ldots,\log(c_d),\log(\rho_d)$ come from multiplicative relations between $c_1,\rho_1,\ldots,c_d,\rho_d$, thus explicitly giving us the required polynomial ideal.
Making sure that Schanuel's conjecture is only needed for termination and not for correctness is non-trivial.
(This issue is explored in some detail in \cite{karimov2025algorithmic}.)
However, no unconditional method is known for computing all polynomial relations between logarithms of algebraic numbers.
Hence, for example, the MSO theory of $\langle \nat; <, {a_1}^\nat, \ldots, {a_d}^\nat \rangle$ is decidable for any positive integers $a_1,\ldots,a_d$, but with a caveat: we know that a decision procedure exists, but to write it down we need to solve a difficult problem about polynomial relations between logarithms of algebraic numbers.

Can we drop the condition (B) in the statement of \Cref{thm::main-ideal-version-consequence}?
If we could, then applying \Cref{thm::sign-pattern-to-prob-1} we would obtain that the MSO theory of the sign pattern of any LRS $\seq{v_n}$ is decidable.
For diagonalisable LRS, this is indeed the case, but with a caveat similar to the above: the sign pattern of a diagonalisable LRS has a toric suffix \cite[Sec.~3]{almagor2021deciding}, and hence a decidable MSO theory~\cite{berthe2023monadic}.
However, no method is known for determining this suffix efficiently.
For non-diagonalisable LRS, essentially nothing is known about the decidability of the MSO theory of the sign pattern.

We will next prove our main theorems (\Cref{powers-main-structure-parametrised}, \Cref{thm::main-schanuel}, \Cref{thm::main-ideal-version}) together.
In the remainder of \Cref{sec::lrs-with-one-dominant-root}, we denote by $P_1,\ldots,P_d$ the respective value sets of integer linear recurrence sequences $\seq{u^{(1)}_n},\ldots, \seq{u^{(d)}_n}$ satisfying the conditions (A-B), by $\varphi$ an MSO formula and by $\Acal$ the corresponding automaton, by $\alpha$ the characteristic word of $P_1,\ldots,P_d$, and by $\beta$ the order word of $P_1,\ldots,P_d$.
Recall from \Cref{thm:MSOautomaton} that deciding whether $\langle \nat; <, P_1,\ldots,P_d\rangle \models \varphi$ is equivalent to deciding whether $\alpha \in L(\Acal)$.

\subsection{From characteristic words to order words}

In this section, we fix LRS $\seq{u^{(1)}_n},\ldots, \seq{u^{(d)}_n}$ satisfying the conditions (A-B) and use the notation above. 
Our approach to deciding whether an automaton $\Acal$ accepts $\alpha$ is to first reduce this to the problem of checking whether an automaton $\Bcal$ accepts the order word $\beta$.
We will do this by showing that the predicates $P_1,\ldots,P_d$ are effectively procyclic and effectively sparse, and then invoking \Cref{thm::sparse-procyclic-reduction-to-order-word}.
Thereafter, we will reduce from the order word to a uniformly recurrent word obtained by ordering the dominant parts of our recurrence sequences.
For this we will need order words obtained from sequences that do not
necessarily take integer values.

\begin{definition}
	\label{def::order-word-for-sequences}
	Let $f_1(n),\ldots,f_d(n) \st \nat \to \rel$ be strictly increasing functions such that $f_i(n) \to +\infty$ as $n \to \infty$ for all~$i$.
	Write $T = \bigcup_{i=1}^d \{f_i(n) \st n \in \nat\}$, and order the terms of $T$ as $\seq{t_n}$.
	The $n$th letter $(b_1,\ldots,b_d)$ of the order word
	\[
	\gamma \coloneqq \operatorname{Ord}(f_1(n), \ldots, f_d(n)) \in (\{0,1\}^d)^\omega
	\]
	 is defined by $b_i = 1 \Leftrightarrow t_n = f_i(m)$ for some $m$.
\end{definition}

The following two lemmas contain the main number-theoretic arguments that we will need.

\begin{lemma}
	\label{lem::lrs-to-exp-baker}
	We can compute $N \ge 0$ with the following properties.
	For $1 \le i \le d$, let $m_i$ be the smallest integer with the property that $u^{(i)}_{m_i} \ge N$.
	\begin{enumerate}
		\item[(a)] For all $i$, $\seq{u^{(i)}_{m_i+n}}$ is strictly increasing.
		\item[(b)] For all $i \ne j$ and $n_i,n_j$ such that $n_i \ge m_i$, we have that $c_i\rho_i^{n_i} \ne c_j\rho_j^{n_j}$, $u^{(i)}_{n_i} \ne u^{(j)}_{n_j}$, and $c_i\rho_i^{n_i} > c_j\rho_j^{n_j} \Leftrightarrow u^{(i)}_{n_i} > u^{(j)}_{n_j}$.
		\item[(c)] For all $i \ne j$ and $n_i \ge m_i$, there exists $n_j \ge m_j$ such that $c_i \rho_i^{n_i} \le c_j \rho_j^{n_j} < \rho_j \cdot c_i \rho_i^{n_i}$.
	\end{enumerate}
\end{lemma}
\begin{proof}
	Denote by $\seq{v^{(i)}_n}$ the non-dominant part of $\seq{u^{(i)}_n}$, and recall that $|v^{(i)}_n| = o((\rho_i^{1-\varepsilon})^n)$ for all sufficiently small $\varepsilon > 0$ (where the implied constant is effective due to \Cref{lem:growth LRS}).
	For all $i$, we have that
	\begin{equation}
		\label{eq::baker-application-1}
		u^{(i)}_{n+1} - u^{(i)}_n = c_i(\rho_i-1)\rho_i^n + v^{(i)}_{n+1} - v^{(i)}_{n}.
	\end{equation}
	Hence we can compute $N_i$ such that for all $m$ with $u^{(i)}_m \ge N_i$, the sequence $\seq{u^{(i)}_{m+n}}$ is strictly increasing.
	We then take $\widetilde{N}_1 = \max_i N_i$.

	We move on to (b) and (c).
	Compute for every pair $i \ne j$, an integer $N_{i,j}$ such that $c_i \rho_i^n \ne c_j \rho_j^m$ for all $n,m$ satisfying $u^{(i)}_n \ge N_{i,j}$.
	To do this, we first compute a basis $X$ of $G_M(\rho_i, \rho_j, c_i/c_j)$ (see \Cref{sec::alg-numbers}).
	If $X = \{(0,0,0)\}$, then $c_i \rho_i^n = c_j \rho_j^m$ has no non-trivial solution in $n,m \in \intg$, and we can take $N_{i,j} = 1$.
	Otherwise, by assumption~(B), we will have $X = \{(k_1,k_2,k_3)\}$ with $k_3\ne 0$.
	If $|k_3| = 1$, then $c_i \rho_i^n = c_j \rho_j^m$ has a unique solution in $n,m \in \intg$, and we can compute $N_{i,j}$ accordingly.
	Otherwise, there are no solutions and we can take $N_{i,j} = 1$.
	We then define $\widetilde{N}_2 = \max \, \{\widetilde{N}_1, \max_{i, j} N_{i,j} \}$.

	Next, applying \Cref{thm::baker-sum-of-two-powers}, we have that for all $i \ne j$ and $n,m \in \nat$ with $\rho_i \ge \rho_j$ and $u^{(i)}_n \ge \widetilde{N}_2$,
	\begin{equation}
		\label{eq::baker-application-1.5}
		|c_i \rho_i^n - c_j \rho_j^m| = |c_j \rho_j^m - c_i \rho_i^n| > \frac{\rho_i^n}{(n+2)^C}
	\end{equation}
	for a computable constant $C$.
	Since
	\begin{equation}
		\label{eq::baker-application-2}
		u^{(i)}_n - u^{(j)}_m = c_i \rho_i^n - c_j \rho_j^m + v^{(i)}_n - v^{(j)}_m
	\end{equation}
	we can compute $M_{i,j} \ge \widetilde{N}_2$ such that, assuming $u^{(i)}_n \ge M_{i,j}$, $u^{(i)}_n - u^{(j)}_m$ is non-zero and has the same sign as $c_i \rho_i^n - c_j \rho_j^m$.
	Let $\widetilde{N}_3 = \max_{i,j} M_{i,j}$.
	
	Finally, we choose $N > \widetilde{N}_3$ such that for all $1 \le i \le d$, there exists $n$ satisfying $N_3 \le u^{(i)}_n < N$.
	Then (a-b) are satisfied since $N \ge \widetilde{N}_3, \widetilde{N}_2$.
	To see that~(c) is satisfied, consider $i \ne j$ and $n_i \ge m_i$.
	Let $n_j$ be the unique \emph{integer} with the property that $c_j \rho^{n_j} \in [c_i \rho_i^{n_i}, \rho_j \cdot c_i \rho_i^{n_i})$.
	We have to argue that $n_j \in \nat$.
	By construction of $N$, there exists $k_j$ such that $N_3 \le u^{(j)}_{k_j} \le N$.
	Because $u^{(j)}_{k_j} \ge \widetilde{N}_3$, we have that $u^{(j)}_{k_j} - u^{(j)}_{n_j}$ has the same sign as $c_j \rho_j^{k_j} - c_j \rho_j^{n_j}$, and the latter is non-zero.
	Since $u^{(j)}_{n_j} > u^{(j)}_{k_j}$, it follows that $n_j > k_j$ and hence $n_j \in \nat$.
\end{proof}

We can now proceed with the proof of procyclicity and sparsity of $P_1,\ldots,P_d$.

\begin{lemma}
	Given $1\le i,j \le d$ and $K \ge 0$, we can compute $L$ such that for all $u^{(i)}_n, u^{(j)}_m \ge L$,
	\[
	u^{(i)}_n \ne u^{(j)}_m \Rightarrow |u^{(i)}_n - u^{(j)}_m| > K.
	\]
\end{lemma}
\begin{proof}
	If $i = j$, we use \Cref{eq::baker-application-1} and the fact that $|v^{(i)}_n| = o((\rho_i^{1-\varepsilon})^n)$ for all sufficiently small $\varepsilon > 0$ with an effective implied constant.
	Suppose $i \ne j$ and, without loss of generality, that $\rho_i \ge \rho_j$.
	Let $N, m_1,\ldots,m_d$ be as in the previous lemma.
	We first enumerate all solutions of $|u^{(i)}_n - u^{(j)}_m| \le K$ with $u^{(i)}_n < N$ using \Cref{lem::lrs-to-exp-baker}~(a).
	For $u^{(i)}_n \ge N$, by \Cref{eq::baker-application-1.5,eq::baker-application-2} we have that
	\[
	|u^{(i)}_n - u^{(j)}_m| = \frac{\rho_i^n}{(n+2)^{C}} + o(\rho_i^{n(1-\varepsilon)})
	\]
	for all sufficiently small $\varepsilon > 0$ (where the implied constant is effective).
	It remains to compute $\nu$ such that the right-hand side is greater than $K$ for all $n \ge \nu$, and then $L \ge N$ such that $u^{(i)}_n \ge L \Rightarrow n \ge \nu$.
\end{proof}

\begin{corollary}
	The predicates $P_1,\ldots,P_d$ are effectively procyclic and effectively sparse.
\end{corollary}
\begin{proof}
	Enumerate the elements of $P_i$ as $\seq{p^{(i)}_n}$, and let $N, m_1,\ldots,m_d$ be as in the statement of \Cref{lem::lrs-to-exp-baker}.
	By \Cref{lem::lrs-to-exp-baker}~(a-b), we have that a suffix of $\seq{p^{(i)}_n}$ (which can be effectively determined) is equal to $\seq{u^{(i)}_{m_i + n}}$, which is effectively procyclic (\Cref{lem:periodic LRS}).
	Hence $P_i$ is effectively procyclic, which implies that $P_1,\ldots,P_d$ are effectively procyclic.
	That $P_1,\ldots,P_d$ are effectively sparse follows immediately from the preceding lemma.
\end{proof}
Recall that $\alpha$ and $\beta$ denote the characteristic and order words of $P_1,\dots,P_d$, respectively.
\begin{corollary}
	\label{thm::char-to-order-of-dominant-parts}
	Given an automaton $\Acal$, we can compute
	\begin{itemize}
		\item an automaton $\Bcal$ such that
		\[
		\alpha \in L(\Acal) \Leftrightarrow \beta \in L(\Bcal),
		\]
		\item and an automaton $\Ccal$ such that
		\[
		\alpha \in L(\Acal) \Leftrightarrow \order{{r_1\rho_1^n}, \ldots, {r_d\rho_d^n}} \in L(\Ccal)
		\]
		where $r_i = c_i \rho_i^{m_i}$.
	\end{itemize}
\end{corollary}
\begin{proof}
	The first claim follows from \Cref{cor::sparse-char-eq-order}.
	To prove the second claim, let $\gamma$ be the word obtained by deleting all occurrences of $\zerovec$ from $\alpha[N,\infty)$.
	Then $\gamma = \beta[\widetilde{N}, \infty)$ for some $\widetilde{N}$.
	Moreover, by  \Cref{lem::lrs-to-exp-baker}~(a-b), 
	\[
	\gamma = \order{{r_1\rho_1^n},\ldots,{r_d\rho_d^n}}.
	\]
	Therefore, $\Ccal$ can be constructed from $\Bcal$ by changing the initial state to the one obtained after $\Bcal$ reads $\beta[0, \widetilde{N})$.
\end{proof}

Henceforth, when the context is clear, we denote by $N$ the smallest integer satisfying the statement of \Cref{lem::lrs-to-exp-baker}, and define $m_1,\ldots,m_d$ as in the statement of \Cref{lem::lrs-to-exp-baker}.
We further write $r_i = c_i \rho_i^{m_i}$, $\gamma = \order{\seq{r_1\rho_1^n}, \ldots, \seq{r_d\rho_d^n}}$, and denote by $\Ccal$ the automaton given in \Cref{thm::char-to-order-of-dominant-parts}.
Note that by \Cref{lem::lrs-to-exp-baker}, each letter in $\gamma$ is of the form $(b_1,\ldots,b_d)$, where exactly one $b_i = 1$ and $b_j = 0$ for all $i \ne j$.
We replace every such letter with $i$ in both $\gamma$ and $\Ccal$ to construct $\widetilde{\gamma} \in \{1,\ldots,d\}^\omega$ and an automaton $\widetilde{\Ccal}$ such that $\gamma \in L(\Ccal) \Leftrightarrow \widetilde{\gamma} \in L(\widetilde{\Ccal})$.

\subsection{Interlude: applying the theory of cutting sequences}
\label{sec::cutting-sequences}
So far we have that $\alpha \in L(\Acal) \Leftrightarrow \beta \in L(\Bcal) \Leftrightarrow \gamma \in L(\Ccal) \Leftrightarrow \widetilde{\gamma} \in L(\widetilde{\Ccal})$, where $\gamma$ is a suffix of $\beta$ and $\widetilde{\gamma}$ is obtained from $\gamma$ through a renaming of letters.
It turns out that $\widetilde{\gamma}$ belongs to the class of \emph{cutting sequences} (also known as \emph{billiard words}), which have been studied in word combinatorics \cite{arnoux1994complexity,baryshnikov1995complexity,bedaride2009directional}.
We illustrate this connection through an example.
\begin{example}
	\label{example-toric1}
	Consider $c_1 = 3, \rho_1 = 2, c_2 = 10, \rho_2 = 3$, $u^{(1)}_n = 3 \cdot 2^n$ and $u^{(2)}_n = 10 \cdot 3^n$.
	We have that $\textcolor{red}{3 \cdot 2^0} < \textcolor{red}{3 \cdot 2^1} < \textcolor{cyan}{10 \cdot 3^0} < \textcolor{red}{3 \cdot 2^2} < \cdots$ and hence the order word is
	\[
	\beta = 
	 \textcolor{red}{(1,0)} \, \textcolor{red}{(1,0)}  \, \textcolor{cyan}{(0,1)} \, \textcolor{red}{(1,0)} \, \textcolor{red}{(1,0)} \, \textcolor{cyan}{(0,1)} \, \textcolor{red}{(1,0)} \, \textcolor{cyan}{(0,1)} \,\textcolor{red}{(1,0)}  \cdots.
	\]
	The smallest $N$ satisfying \Cref{lem::lrs-to-exp-baker} (particularly note the property~(c)) yields $m_1 = 1$, $m_2 = 0$, and $\gamma = \beta[1, \infty)$.
	Performing the  renaming $(1,0) \to \textcolor{red}{1}$, $(0,1) \to \textcolor{cyan}{2}$, we obtain
	$\widetilde{\gamma} = \textcolor{red}{1}  \textcolor{cyan}{2} \textcolor{red}{1} \textcolor{red}{1} \textcolor{cyan}{2} \textcolor{red}{1} \textcolor{cyan}{2} \textcolor{red}{1}  \cdots$.
	Let $a_n = \log(c_1\rho_1^n) = \log(3) + n\log(2)$ and $b_n = \log(c_2\rho_2^n) =  \log(7)+n\log(3)$.
	Figure~\ref{fig::cutting} (left) illustrates a way to generate $\beta$ (up to the renaming of letters above).
	We start at the point $(0,0)$ and follow the line $y = x$.
	Every time a vertical line $x = a_n$ for some~$n$ is hit, we write~\textcolor{red}{$(1,0)$}.
	When we hit a horizontal line $y = b_n$ for some~$n$, we write \textcolor{cyan}{$(0,1)$}.
	If we discard the first letter of $\beta$, we obtain a cutting sequence, illustrated in Fig.~\ref{fig::cutting} (right).
	The figure on the right is obtained from the one on the left by a translation and a scaling.
	In the former, we start at a point $(0,y)$, where $0 < y < 1$, and follow the dashed line, which has the slope $\log(2)/\log(3)$.
	When we hit a line $x = n$ for $n \in \nat$, we write $\textcolor{red}{1}$;
	When we hit $y = n$, we write~$\textcolor{cyan}{2}$.
\end{example}

\begin{figure}
	\begin{subfigure}[b]{0.48\textwidth}
		\centering
		\begin{tikzpicture}
			\def\size{5.1}
			\draw[thick, ->] (0,0) -- (0,\size);
			\draw[thick, ->] (0,0) -- (\size, 0);
			\draw[dashed] (0,0) -- (\size-0.03,\size-0.03);

			\foreach \x in {0,...,5}
			{
				\draw[very thin]  (\x * \logtwo + \logthree, 0) -- (\x * \logtwo + \logthree, \size);
				\draw[ultra thick,red] (\x * \logtwo + \logthree, \x * \logtwo-0.12+\logthree) -- (\x * \logtwo+\logthree,\x * \logtwo+0.12+\logthree);
				\node at (\x * \logtwo + \logthree, -0.3) {\large$a_{\x}$};
			}
			
			\foreach \x  in {0,...,2}
			{
				\draw[very thin] (0, \logten + \x * \logthree ) -- (\size, \logten + \x * \logthree );
				\draw[ultra thick,cyan] 
				(\logten + \x * \logthree-0.12, \logten + \x * \logthree) 
				-- 
				( \logten + \x * \logthree+0.12, \logten + \x * \logthree);
				\node at (-0.3, \logten + \x * \logthree) {\large$b_{\x}$};
			}
		\end{tikzpicture}
	\end{subfigure}
	\begin{subfigure}[b]{0.48\textwidth}
		\centering
		\begin{tikzpicture}
			\def\size{4.8};
			\def\k{\logtwo/\logthree};
			\def\sizey{\size*\k+1.5};
			\def\b{0.51};

			\draw[thick, ->] (0,0) -- (0,\sizey);
			\draw[thick, ->] (0,0) -- (\size, 0);
			\draw[dashed] (0,\b) -- (\size, {\k*\size + \b});

			\foreach \x in {1,...,4}
			{
				\draw[very thin] (\x,0) -- (\x,\sizey);
				\draw[ultra thick,red] (\x,  {\k*\x+\b - 0.12}) -- (\x,  {\k*\x+\b+0.12});
				\node at (\x, -0.3) {\large$\x$};
			};
			\draw[ultra thick,red] (0,  {\b - 0.12}) -- (0,  {\b+0.12});

			\foreach \y  in {1,...,3}
			{
				\draw[very thin] (0, \y) -- (\size, \y);
				\draw[ultra thick,cyan] ({-0.12 + \y * 1.585 - \b * 1.585}, \y) -- ({+0.12 + \y * 1.585 - \b * 1.585}, \y);
				\node at (-0.3, \y) {\large$\y$};
			};
			\draw[very thin] (0, 4) -- (\size, 4);
			\node at (-0.3, 4) {\large$4$};

		\end{tikzpicture}
	\end{subfigure}
	\caption{Generating the order word $\beta$ (left) and its suffix (up to a renaming) $\widetilde{\gamma}$, which is a cutting sequence (right).}
	\label{fig::cutting}
\end{figure}

Formally, the cutting sequence over $\Sigma = \{1,\ldots,d\}$ with slopes $\lambda_1,\ldots,\lambda_d$ and intercepts $s_1,\ldots,s_d$ is the infinite word obtained by considering intersections of the line $\{(\xi_1+\lambda_1t, \ldots, \xi_d+\lambda_dt) \st t \ge 0\}$ with the grid lines $x_i = c$ (where $c \in \nat$ and $x_1,\ldots,x_d$ denote the $d$ standard coordinates in $\rel^d$) as $t \to \infty$ starting from $t = 0$, assuming no two grid lines are intersected simultaneously.
Cutting sequences are uniformly recurrent, and under some mild assumptions on the slope, there is an explicit formula for the factor complexity, which allows us to apply \Cref{thm::semenov-2}.

\begin{itemize}
	\item[(i)] If $d = 2$ and $\log(\rho_1)/\log(\rho_2)$ is irrational (which is the case in Fig.~\ref{fig::cutting}), then $\gamma$ is a \emph{Sturmian word} and therefore $\pi_\gamma(n) =n +1$. See, e.g.\, \cite[Sec.~10.5]{allouche2003automatic} and \cite[Chap.~2]{lothaire2002algebraic}.
	\item[(ii)] By \cite{arnoux1994complexity}, if $d = 3$, and $1/\log(\rho_1),1/\log(\rho_2),1/\log(\rho_3)$ as well as $\log(\rho_1)$, $\log(\rho_2),\log(\rho_3)$ are linearly independent over $\rat$, then $\pi_\gamma(n) = n^2 + n +1$.
	\item[(iii)] For arbitrary $d > 0$, B\'edaride \cite{bedaride2009directional} gives an exact formula for $\pi_\gamma(n)$ assuming $1/\log(\rho_1),\ldots,1/\log(\rho_d)$ as well as every triple $\log(\rho_i), \log(\rho_j), \log(\rho_k)$ for pairwise distinct $i,j,k$ are linearly independent over $\rat$.
	This generalises the well-known result \cite{baryshnikov1995complexity} of Baryshnikov, which gives an exact formula for the factor complexity assuming both the logarithms and their inverses are linearly independent over $\rat$.
\end{itemize}

Note that the exact value of the intercept does not matter in (i-iii) above:
the only requirement is that no two grid hyperplanes be simultaneously reachable
We will not be using cutting sequences to prove our main results.
However, we could prove weaker results using the theory of cutting sequences:
it is not difficult to prove that $\widetilde{\gamma}$ is the cutting sequence generated by the line $\{(\chi_1+t/\log(\rho_1), \ldots, \chi_d+t/\log(\rho_d)) \st t \ge 0\}$, where $\chi_i \in [0,1)$ for all $i$.
The following is an immediate consequence of \Cref{thm::semenov-2} and the aforementioned result of B\'edaride.

\begin{proposition}
	Problem 1 is decidable for integer LRS  $\seq{u^{(1)}_n},\ldots, \seq{u^{(d)}_n}$ satisfying the conditions (A), (C), and that every triple $\log(\rho_i), \log(\rho_j), \log(\rho_k)$ for pairwise distinct $i,j,k$ are linearly independent over $\rat$.
\end{proposition}

Note that, however, this is strictly weaker than \Cref{powers-main-structure-parametrised}.
Consider, for example, $\rho_1 = 2$, $\rho_2= 3$ and $\rho_3= 6$.
By \Cref{lem::rank-d-2-to-LI}, $1/\log(\rho_1), 1/\log(\rho_2),1/\log(\rho_3)$ are linearly independent over~$\rat$, but $\log(\rho_1),\log(\rho_2),\log(\rho_3)$ are not.

\subsection{Deciding whether $\widetilde{\gamma} \in L(\widetilde{\Ccal})$}
So far, we have only used the assumptions (A-B) to reduce the problem of deciding whether $\alpha \in L(\Acal)$ to whether  $\widetilde{\gamma} \in L(\widetilde{\Ccal})$.
We will show that (i) $\widetilde{\gamma}$ is uniformly recurrent and (ii) apply \Cref{thm::semenov-2} to decide whether $\widetilde{\Ccal}$ accepts $\widetilde{\gamma}$.
Let $\Sigma = \{1,\ldots,d\}$, and recall that $\widetilde{\gamma} \in \Sigma^\omega$ and $\widetilde{\gamma} = \order{\seq{r_1\rho_1^n}, \ldots, \seq{r_d\rho_d^n}}$ up to the renaming of letters that maps the letter $(b_1,\ldots,b_d)$ with $b_i = 1$ and $b_j = 0$ for $j \ne i$ to the letter $i$.
For $b \in \Sigma$ and $w \in \Sigma^*$, we denote by $|w|_b$ the number of occurrences of $b$ in $w$.

Let $w = b_0 b_1 \cdots b_m \in \Sigma^*$ and $b = b_0$.
Define $f(n, i)$ to be the smallest term $r_i \rho_i^m$ such that $m \ge 0$ and $r_i \rho_i^m \ge r_b \rho_b^n$, and let $f(n, i, k) = \rho_i^{k-1} f(n, i)$.
Thus, $f(n,i,k)$ is the $k$th largest term of the form $r_i \rho_i^m$, such that $m \ge 0$ and $r_i \rho_i^m \ge r_b \rho_b^n$, counting from one.
Further write $\nu_i(n) = f(n, b_i, |w[0,i+1)|_{b_i})$ and $\tau_i(n) = f(n, b_i, |w|_{b_i}+1)$.
Then the word $w$ occurs in $\widetilde{\gamma}$ at the position corresponding to $r_b \rho_b^n$ if and only if
\begin{equation}
	\label{eq::pattern-occurrence-1}
	\nu_0(n) < \cdots < \nu_m(n) < \tau_1(n), \tau_2(n), \ldots, \tau_d(n)
\end{equation}
which is equivalent to
\begin{equation}
	\label{eq::pattern-occurrence-1.5}
	\frac{r_b \rho_b^n}{\nu_0(n)} > \cdots > \frac{r_b \rho_b^n}{\nu_m(n)} > \frac{r_b \rho_b^n}{\tau_1(n)}, \frac{r_b \rho_b^n}{\tau_2(n)}, \ldots, \frac{r_b \rho_b^n}{\tau_d(n)}.
\end{equation}
Next, observe that by definition and \Cref{lem::lrs-to-exp-baker}~(c),
\[
\frac{1}{\rho_i} \le \frac{r_b \rho_b^n}{f(n, i)} < 1
\]
for all $i$, which implies that
\begin{equation}
	\label{eq::pattern-occurrence-1.7}
	\frac{1}{\rho_i^{k+1}} \le \frac{r_b \rho_b^n}{f(n, i, k)} < \frac{1}{\rho_i^k}
\end{equation}
for all $i, k$.

Write $f(n,i,k) = r_i \rho_i^l$.
Taking logarithms, \eqref{eq::pattern-occurrence-1.7} is equivalent to
\begin{equation}
	\label{eq::pattern-occurrence-1.8}
	- (k+1) \log(\rho_i) \le
	\log
	\bigg(
	\frac{r_b \rho_b^n}{f(n, i, k)}
	\bigg) = \log \bigg( \frac{r_b}{r_i} \bigg) + n \log(\rho_b) - l \log(\rho_i)
	< - k\log(\rho_i)
\end{equation}
which is equivalent to
\begin{equation}
	\label{eq::pattern-occurrence-1.9}
	0 \le \log \bigg(\frac{r_b}{r_i}\bigg) + n \log(\rho_b) + (k+1-l)\log(\rho_i) < \log(\rho_i).
\end{equation}
Dividing by $\log(\rho_i)$, we obtain that
\begin{equation}
	\label{eq::pattern-occurrence-1.95}
	0 \le \frac{\log(r_b) - \log (r_1)}{\log(\rho_i)} + n \frac{\log(\rho_b)}{\log(\rho_i)} + (k + 1 - l) < 1.
\end{equation}
Define $s_i =  \frac{\log r_b - \log r_i}{\log(\rho_i)}$, $\delta_i = \frac{\log \rho_{b}}{\log \rho_i}$, $s = (s_1,\ldots,s_d)$, and $\delta = (\delta_1,\ldots,\delta_d)$.
Then from \eqref{eq::pattern-occurrence-1.95} we conclude that
\[
\{\gamma_i + n\delta_i\} = \frac{\log(r_b) - \log (r_1)}{\log(\rho_i)} + n \frac{\log(\rho_b)}{\log(\rho_i)} + (k + 1 - l).
\]
From \eqref{eq::pattern-occurrence-1.8} it then follows that
\[
\frac{r_b \rho_b^n}{f(n, i, k)} =
	\exp\big(
	-
	(k+1)\log(\rho_i)
	+
		\log(\rho_i)
		\{
		\gamma_i + n \delta_i
		\}
	\big).
\]
Hence, taking logarithms, we can write \eqref{eq::pattern-occurrence-1.5} as
\begin{multline}
	\label{eq::pattern-occurrence-2}
	k_1\log(\rho_{b_0}) +
	\log(\rho_{b_0})
	\big\{
	\gamma_{b_0}+ n \delta_{b_0}
	\big\}
	>
	\cdots
	>
	k_{m} \log(\rho_{b_m}) +
	\log(\rho_{b_m})
	\big\{
	\gamma_{b_m}+ n \delta_{b_m}
	\big\}
	>
	\\
	l_{1}\log(\rho_{1}) +
	\log(\rho_{1})
	\big\{
	\gamma_1 + n \delta_1
	\big\},
	\ldots,
	l_d\log(\rho_{d}) +
	\log(\rho_{d})
	\big\{
	\gamma_d + n \delta_d
	\big\}
\end{multline}
where $k_i, l_i$ are integers.

So, what did we achieve?
Let $O_w$ be the open subset of $\torus^d$ consisting of all $x = (x_1,\ldots,x_d)$ such that
\begin{multline}
	\label{eq::pattern-occurrence-3}
	k_1\log(\rho_{b_0}) +
	\log(\rho_{b_0})
	x_{b_0}
	>
	\cdots
	>
	k_{m} \log(\rho_{b_m}) +
	\log(\rho_{b_m})
	x_{b_m}
	>
	\\
	l_{1}\log(\rho_{1}) +
	\log(\rho_{1})
	x_1,
	\ldots,
	l_d\log(\rho_{d}) +
	\log(\rho_{d})
	x_d.
\end{multline}
We have the toric dynamical system (see \Cref{sec::toric-words}) given by $x \mapsto x + \delta$, and whether the pattern $w$ occurs at a position corresponding to $r_b \rho_b^n$ is characterised by whether $\{s + n\delta\} \in O_w$, i.e.\ whether the orbit of $s$ falls into the open set~$O_w$ at the step $n$.
By uniform recurrence (\Cref{thm::ur-on-torus}), the set of all $n \in \nat$ such that  $\{s + n \delta\} \in O_w$ is either empty or is infinite and has bounded gaps.
Because the letter $b_0$ occurs infinitely often in $\widetilde{\gamma}$  and with bounded gaps (recall that $\widetilde{\gamma} = \order{\seq{r_1\rho_1^n}, \ldots, \seq{r_d\rho_d^n}}$ up to a renaming of letters), it follows that either $w$ does not occur in $\widetilde{\gamma}$ at all, or does so infinitely often and with bounded gaps.
\textbf{Therefore, $\widetilde{\gamma}$ is uniformly recurrent.}

We now prove all of our main theorems.
Let $w = b_0 b_1 \cdots b_m \in \Sigma^*$ and $b = b_0$, and construct $\delta = (\delta_1,\ldots,\delta_d)$ and $s = (s_1,\ldots,s_d)$ as above.
Recall that by \Cref{thm::semenov-2}, to decide whether $\widetilde{\gamma} \in L(\widetilde{\Ccal})$ it suffices to be able to effectively check whether $w$ occurs in $\widetilde{\gamma}$.
Let $X_\delta$ be the closure of $\seq{\{n \delta\}}$, which is a subset of $\torus^d$ defined by $\intg$-affine equalities; see \Cref{sec::toric-words}.
Further let $X_{\delta,s} = s + X_\delta$; this is the closure of $\seq{\{s+n \delta\}}$.
We have that the word $w$ occurs (infinitely often) in $\widetilde{\gamma}$ if and only if $X_{\delta,s}$ intersects~$O_w$.

\begin{proof}[Proof of \Cref{powers-main-structure-parametrised}]
	Suppose $\frac{1}{\log(\rho_1)},\ldots, \frac{1}{\log(\rho_d)}$ are linearly independent.
	We will show how to decide whether $w$ occurs in $\widetilde{\gamma}$.
	Consider the equation
	\[
	z_1 \delta_1 + \cdots  + z_d\delta_d = z_0
	\]
	where $z_i \in \intg$.
	Because $\delta_b = 1$, it is equivalent to
	\[
	\sum_{i \ne b} \frac{z_i}{\log(\rho_i)} = \frac{a_0-z_b}{\log(\rho_b)}.
	\]
	By the linear independence assumption, the solutions are
	\[
	z_1 = \cdots = z_{b-1} = z_{b+1} = \cdots = z_d = z_0 - z_b = 0.
	\]
	From Kronecker's theorem (\Cref{sec::toric-words}) it follows that 
	\[
	X_\delta = \{(x_1,\ldots,x_d) \in \torus^d \st x_b = 0\}.
	\]
	Since $s_b = 0$, we have that $X_{\delta,s} \coloneqq \{y \st y = x + s \textrm{ for some $x \in X_\delta$}\} = X_\delta$.
	Hence $w$ occurs (infinitely often) in $\widetilde{\gamma}$ if and only if $O_w \cap X_\delta$ is non-empty.
	To check this condition, we perform a change of variables $y_i = \log(\rho_i)x_i$ on the system of equations defining $O_w$ to obtain
	\[
	\begin{cases*}
		\, 0 \le y_i < \log(\rho_i) \: \textrm{ for $1\le i\le d$}, \;\: i \ne b\\
		\, y_b = 0 \\
		\, k_1\log(\rho_{b_0}) +
		y_{b_0}
		>
		\cdots
		>
		k_{m} \log(\rho_{b_m}) +
		y_{b_m}
		>
		l_i \log(\rho_i)  \: \textrm{ for $i = 1,\ldots,d$.}
	\end{cases*}
	\]
	Note that the coefficients of all $y_1,\ldots,y_d$ are 1.
	Eliminating $y_1,\ldots,y_d$ using the Fourier-Motzkin algorithm (\Cref{sec::fourier-motzkin}), we transform the system above into a Boolean combination of linear inequalities in $\log(\rho_1),\ldots,\log(\rho_d)$, which we then solve using \Cref{lem::baker-determining-sign}.
\end{proof}

\begin{proof}[Proof of \Cref{thm::main-ideal-version}]
	Now suppose we have access to a basis of $G_L(1/\log(\rho_1),\ldots,1/\log(\rho_d))$ and the ideal of all polynomial relations between $\log(c_1),\log(\rho_1),\ldots,\log(c_d),\log(\rho_d)$.
	We will show how to decide whether $w$ occurs in $\widetilde{\gamma}$, or equivalently, whether $X_{\delta,s}$ intersects $O_w$.
	As discussed above, for all $z_0,\ldots,z_d\in\intg$,
	\[
	\sum_{i=1} z_i \delta_i = z_0 \Leftrightarrow \sum_{i\ne b} \frac{z_i}{\log(\rho_i)} = \frac{z_0-z_b}{\log(\rho_b)}.
	\]
	Therefore, from a basis of $G_L(1/\log(\rho_1),\ldots,1/\log(\rho_d))$ we can compute a basis of $G_A(\delta_1,\ldots,\delta_d)$ and hence a conjunction of $\intg$-linear equalities defining $X_\delta \subseteq \torus^d$ (see \Cref{sec::toric-words}).
	Hence whether $X_{\delta,s} = s + X_\delta$ intersects $O_w$ can be written as the system consisting of \Cref{eq::pattern-occurrence-3}, the equation $0 \le x_1,\ldots,x_d < 1$, and $\intg$-affine equations (in $x_i -s_i$ for $1\le i \le d$) stating that $(x_1-s_1,\ldots,x_d-s_d) \in X_{\delta}$ obtained from the basis of $G_A(\delta_1,\ldots,\delta_d)$.
	Recall that $\log(r_i)$, which appears in the definition of $s_i$, is equal to $\log (c_i \rho_i^{m_i}) = \log(c_i) + m_i \log(\rho_i)$.
	We thus have a system of polynomial equalities and inequalities in $\log(c_1),\log(\rho_1),\ldots,\log(c_d),\log(\rho_d)$ that holds if and only if $X_{\delta,s}$ intersects~$O_w$.
	To solve the system, we first check which $p(\log(c_1),\log(\rho_1),\ldots,\log(c_d),\log(\rho_d))$, where $p$ is a polynomial appearing in the system, are equal to zero using the ideal of all polynomial relations (i.e., solving the ideal membership problem using, e.g., Gr\"obner bases \cite{Cox2025}).
	Thereafter, we can determine the signs of non-zero $p(\log(c_1),\log(\rho_1),\ldots,\log(c_d),\log(\rho_d))$ by computing sufficiently close over- and under-approximations thereof.
\end{proof}

\begin{proof}[Proof of \Cref{thm::main-schanuel}]
	It suffices to give an algorithm for deciding whether $w$ occurs in $\widetilde{\gamma}$ that relies on Schanuel's conjecture only for termination.
	We run two semi-algorithms in parallel.
	On one hand, we generate larger and larger prefixes of $\widetilde{\gamma}$, and halt if we detect an occurrence of the finite word $w$.
	On the other hand, we enumerate all finite subsets $V$ of $\rat^{d+1}$.
	For each such $V$, using the decision procedure for the first-order theory of $\rexp$, we check whether for all $v \coloneqq (v_0, \ldots, v_d) \in V$, $v_0 = \sum_{i=1}^d \frac{v_i}{\log(\rho_i)}$.
	That is, we check whether $V$ is a set of affine relations satisfied by $\frac{1}{\log(\rho_1)},\ldots,\frac{1}{\log(\rho_d)}$.
	If the answer is negative, we generate the next $V$.
	If the answer is positive, we compute
	\[
	X_V = \{(x_1,\ldots,x_d) \in \torus^d \st v_0 = x_1v_1 + \cdots + x_dv_d\}
	\]
	which satisfies $X_V \supseteq X_\delta$.
	We then check, again using the decision procedure for $\rexp$, whether $s + X_V$ (which contains $X_{\delta,s}$) intersects $O_w$.
	If yes, we move on to the next value of $V$.
	If no, we have found a certificate that $X_{\delta,s}$ does not intersect $O_w$, i.e.\ the word $w$ does not occur in $\widetilde{\gamma}$.
	Note that this algorithm always terminates (assuming Schanuel's conjecture) since in case $w$ does not occur, eventually we will generate $V$ such that $X_V = X_\delta$.
\end{proof}

\section{MSO decidability via expansions in integer bases}
\label{sec:MSO normal numbers}

In this section, we discuss MSO theories of structures of the form $\langle \nat; <, P_1, P_2 \rangle$ where $P_1 = \{qn^d \st n \in \nat\}$ and $P_2 = \{pb^n \st n \in \nat\}$ for some positive integers $q,p,b,d$.
Note that both $\seq{q n^d}$ and $\seq{ p b^n}$ are linear recurrence sequences with a single dominant root ($1$ and $b$, respectively), but in the former case the dominant root is repeated (in particular, it has multiplicity $d+1$), which does not fall into the scope of the previous section.
If $d=1$ or $b=1$, then at least one of $P_1,P_2$ can be defined in $\langle \nat; < \rangle$, and hence the MSO theory of $\langle \nat; <, P_1, P_2 \rangle$ is decidable by the result of Carton and Thomas \cite{cartonthomas}.
Our approach is again to use automata-theoretic tools to reduce the MSO decision problem to a problem about dynamical systems.
However, in this case the relevant dynamical system is given by $x \mapsto \{b \cdot x\}$ (in contrast to $x \mapsto \{x + t\}$), which generates the greedy expansion $\tau \in \{0,\ldots,b-1\}^\omega$ of $x \in (0,1)$ in base $b$.
Our main result is that the MSO theories of our structures are intimately connected to MSO theories of base-$b$ expansions of certain numbers.
Recall that for any $\Sigma \subseteq \nat$, we can view a word $\tau \in \Sigma^\omega$ as a function of type $\nat \to \Sigma$.

\begin{theorem}
	\label{thm::main poly vs exp}
	Let $b, d \ge 2$ and $p, q \ge 1$ be integers, $P_1 = \{qn^d \st n \in \nat\}$, and $P_2 = \{pb^n \st n \in \nat\}$.
	Write $\eta = \sqrt[d]{p/q}$, $\zeta = \sqrt[d]{1/b}$, and let $\gamma_0,\ldots,\gamma_{d-1} \in \{0,\ldots,b-1\}^\omega$ be the base-$b$ expansions of $\{\eta\}, \{\eta \zeta\},\dots,\{\eta\zeta^{d-1}\}$, respectively.
	Then the MSO theories of $\langle \nat; <, P_1,P_2\rangle$ and $\langle \nat; <, \gamma_0,\ldots,\gamma_{d-1} \rangle$ are Turing-equivalent.
\end{theorem}

For example, if $P_1 = \{27n^3 \st n\in\nat\}$ and $P_2 = \powersofk{8}$, then $d = 3$, $\eta = 1/27$, $\zeta = 1/8$, and the base-$8$ expansions of $1/3$, $1/6$, and $1/12$ underlie the pair of predicates $P_1, P_2$.
Recall that the expansion of a rational number in any base is ultimately periodic with a computable period and pre-period (see, e.g., \cite[Thm.~12.4]{rosen2011elementary}), and hence definable in $\langle \nat; < \rangle$.
In our example, because $\eta, \zeta \in \rat$, all three relevant expansions are periodic and the attendant MSO theory is decidable.
Let us formalise a few similar corollaries before proving our main result.

\begin{corollary}
	Assuming \Cref{conj: normal}, the MSO theory of $\langle \nat; <, P_1, P_2 \rangle$ is decidable whenever $d=2$ and at least one of $\eta, \eta\zeta$ is rational.
\end{corollary}
\begin{proof}
	Suppose $\eta \in \rat$ or  $\eta\zeta \in \rat$.
	Note that $\eta, \zeta$ are both algebraic.
	If $\eta, \zeta \in \rat$, then $\gamma_0, \gamma_1$ are periodic and decidability is immediate.
	Otherwise, for some $i \in \{0,1\}$, $\gamma_i$ is periodic and $\gamma_{1-i}$ is disjunctive by \Cref{conj: normal}.
	Decidability then follows from \Cref{normal-decidable}.
\end{proof}

\begin{corollary}
  The MSO theory of $\langle \nat; <, \kthpowers{2}, \powersofk{2}\rangle$ is Turing-equivalent to that of $\langle \nat; <, \gamma\rangle$, where $\gamma$ is the binary expansion of $\sqrt{2}-1$.
\end{corollary}
\begin{proof}
	By \Cref{thm::main poly vs exp}, the MSO theory of $\langle \nat; <, \kthpowers{2}, \powersofk{2}\rangle$ is Turing-equivalent to that of $\langle \nat; <, \gamma_0, \gamma_1 \rangle$ where $\gamma_0(n) = 0$ for all $n$ and $\gamma_1$ is the binary expansion of $\sqrt{\frac 1 2} = \frac{1}{2} \cdot \sqrt{2}$, which is simply a shifted version of the binary expansion of~$\sqrt{2}-1$.
\end{proof}
\begin{corollary}
	\label{thm::main-2-cor-3}
	Let $b,d,p,q, P_1,P_2, \eta, \zeta$ be as in the statement of \Cref{thm::main poly vs exp}. 
	Suppose $1/\zeta =\ell$ for an integer $\ell$.
	Then the (decision problem of the) MSO theory of $\langle \nat; <, P_1,P_2\rangle$ is Turing equivalent to $\Acc{\beta}$, where $\beta$ is the base-$\ell$ expansion of $\eta$.
	In particular, the MSO theory of 
	$\langle \nat; <, \{(k^d)^n \st n \in \nat\}, \kthpowers{d}\rangle$ is decidable for any integers $k,d \ge 2$.
\end{corollary}
\begin{proof}
	We have that $b = \ell^d$.
	Let $\gamma = \gamma_0 \times \cdots \times \gamma_{d-1}$.
	Applying \Cref{thm::main poly vs exp}, it suffices to show that $\Acc{\gamma}$ is Turing-equivalent to $\Acc{\sigma}$.
	
	If $x \in (0,1)$ has the base-$b$ expansion $\sigma \in \{0,\dots,b-1\}^\omega$, then its base-$\ell$ expansion can be obtained by simply replacing each letter $\sigma(n)$ with its base-$\ell$ expansion padded with leading zeros to make the total length equal to $d$.
	Hence we can map the base-$b$ expansion of a number to its base-$\ell$ expansion using a transducer, and vice versa.
	
	Let $\widetilde{\gamma}_i$ be the base-$\ell$ expansion of $\{\eta \zeta^i\} = \{\eta \ell^i\}$, which is just a shift of the base-$\ell$ expansion of $\{\eta\}$.
	Combining this with the earlier argument, we obtain that the base-$b$ expansion $\gamma_i$ of $\{\eta \zeta^i\}$ can be obtained from the base-$\ell$ expansion $\beta$ of~$\{\eta\}$ via a transduction, and vice versa.
	The same conclusion then also holds for $\gamma$ and $\beta$.
	It remains to apply \Cref{cor::transducer reduction}.
	
	To prove the second statement, note that in that case we have $b = k^d$, $1/\zeta = k \in \intg$, $\eta = 1$, and $\beta = 0^\omega$.
\end{proof}

\begin{proof}[Proof of \Cref{thm::main poly vs exp}]
Let $\widetilde{\alpha} \in (\{0,1\}^2)^\omega$ be the characteristic word of $P_1,P_2$.
For $0 \le i < d$, let $Q_i = \{p b^{dn+i} \st n \in \nat\}$, and let $\alpha \in (\{0,1\}^{d+1})^\omega$ be the characteristic word of $P_1, Q_0, \ldots, Q_{d-1}$.
We can construct transducers $\Bcal, \Ccal$ such that $\alpha = \Bcal(\widetilde{\alpha})$ and $\widetilde{\alpha} = \Ccal(\alpha)$. 
By \Cref{thm::transduction}, $\Acc{\widetilde{\alpha}}$ and $\Acc{\alpha}$ are Turing-equivalent.
Therefore, our goal is to show that $\Acc{\alpha}$ is Turing-equivalent to $\Acc{\gamma}$, where $\gamma = \gamma_0 \times \cdots \times \gamma_{d-1} \in (\{0,\ldots,b-1\}^d)^\omega$.

Let $\beta  \in (\{0,1\}^{d+1})^\omega$ be the order word of $P_1,Q_0,\ldots,Q_{d-1}$. 
We first show that $P_1,Q_0, \ldots, Q_{d-1}$ are effectively procyclic and effectively sparse, and hence $\Acc{\alpha}$ is Turing-equivalent to $\Acc{\beta}$.
Each predicate is the value set of a strictly increasing integer linear recurrence sequence, and hence is effectively procyclic (\Cref{lem:periodic LRS}).
It remains to argue that they are collectively effectively sparse.
This follows immediately from the facts that 
\begin{itemize}
	\item $|p b^{dn_1 + r_1} - p b^{dn_2+r_2}| \ge p b^{\min \{dn_1+r_1,n_2+r_2\}} \cdot \big(1 - \frac{1}{b}\big)$ for all $n_1,n_2,r_1,r_2$ such that $p b^{dn_1 + r_1} - p b^{dn_2+r_2} \ne 0$,
	\item $\lim_{n \to \infty} q(n+1)^d - qn^d = \infty$, 
\end{itemize}
and the following classical result of Schinzel and Tijdeman~\cite{schinzel}.
 \begin{theorem}
 	\label{spaced-out}
 	For every $N \ge 1$, the equation $|qn^d - pb^{m}| = N$ has finitely many solutions $(n, m) \in \nat^2$ that can be effectively enumerated.
 \end{theorem}

We next prove that $\Acc{\beta}$ is Turing-equivalent to $\Acc{\gamma}$.
Note that $Q_{i} \cap Q_j = \varnothing$ for all $i\ne j$, and
\[
\order{Q_0,\ldots,Q_{d-1}} = ((1,0,\ldots,0) \: (0,1,0,\ldots,0) \: \cdots \: (0,\ldots,0,1))^\omega.
\]
This is because, when constructing $Q_0, \ldots, Q_{d-1}$, we simply took the $d$ alternating subsequences from the strictly increasing sequence $\seq{pb^n}$.
By \Cref{def::order-word-for-sequences}, we have that
\[
\beta = \order{qn^d, {pb^{nd}}, {pb^{nd+1}},\ldots, {pb^{nd+(d-1)}}}.
\]
Dividing by $q$ and taking $d$th roots, 
\[
\beta = \order{ {n}, v^{(0)}_n, \ldots, v^{(d-1)}_n}
\]
where
\begin{equation}\label{eq::base-expansions-3}
	v^{(i)}_n = \sqrt[d]{\frac{p}{qb^{d-i}}} \cdot b^{n+1} = \eta \zeta^{d-i} \cdot b^{n+1}.
\end{equation}
To see the connection to expansions in base $b$, recall that for any $x > 0$, the $n$th digit in the base-$b$ expansion of $\{x\}$ is given by $\floor{x \cdot b^{n+1}} \bmod b$.
Hence
\begin{equation}\label{eq::base-expansions--1}
	\gamma_i(n) = \floor{v^{(i)}_n} \bmod d
\end{equation}
for all $i,n$.
Next, observe that $v^{(i)}_n < v^{(j)}_n$ 
for all $n$ and $i < j$, and $v^{(i)}_n < v^{(j)}_m$ for all $i,j$ and $n < m$.
We can therefore write
\[
\beta = \prod_{n=0}^\infty w_n, \qquad\qquad 
w_n =
(1,0,\ldots,0)^{k^{(1)}_n} z_1 \cdots (1,0,\ldots,0)^{k^{(d)}_n} z_d
\]
where $k^{(i)}_n \ge 0$ and
\[
z_i = (y_i, \underbrace{0, \ldots, 0}_{\textrm{$i-1$ times}}, 1, 0, \ldots, 0)
\] 
with $y_i \in \{0,1\}$ for all $i$ and $n$.
The term $v^{(i)}_n$ corresponds to the letter $z_i$ in the factor $w_n$.
By simple counting,
\begin{equation}\label{eq::base-expansions-0}
	\lfloor v^{(i)}_n \rfloor = | \{ l \in \nat \st l \le v^{(i)}_n\} | = \sum_{m=0}^{n-1} \sum_{j=0}^{d-1} (k^{(j)}_m + y_j) + \sum_{j=0}^i (k^{(j)}_m + y_j)
\end{equation}
for all $n, i$.
Write $v^{(-1)}_{n+1} = v^{(d-1)}_{n}$ for all $n, i$.
We then have that
\begin{equation}\label{eq::base-expansions-1}
	\lfloor v^{(i)}_n \rfloor  = \floor{v^{(i-1)}_{n}} + k^{(i)}_n + y_i.
\end{equation}
or all $n \ge 1$ and $i$.
Note that
\[
\frac{v^{(i)}_n}{v^{(i-1)}_{n}} = b^{\frac{1}{d}} > 1
\]
for all $n \ge 1$ and $i$, and hence
\begin{equation}\label{eq::base-expansions-1.5}
	\lim_{n \to \infty} k_n^{(i)} = \infty
\end{equation}
for all $i$.
On the other hand, $v^{(i)}_{n+1} = b \cdot v^{(i)}_n$ for all $n, i$, and from \Cref{eq::base-expansions--1} it therefore follows that
\begin{equation}\label{eq::base-expansions--2}
	\floor{v^{(i)}_{n+1}} = b \cdot \floor{v^{(i)}_n} + \gamma_i(n).
\end{equation}
We can now prove that $\Acc{\beta}$ and $\Acc{\gamma}$ are Turing-equivalent.

{\bfseries$\Acc{\gamma}$ reduces to $\Acc{\beta}$.} 
We give a transducer $\Bcal$ such that $\gamma = \Bcal(\beta)$, and apply \Cref{thm::transduction}.
The transducer reads $\beta$ in chunks $\seq{w_n}$, and outputs $\gamma(n)$ after reading $w_n$.
Before reading $w_n$ for $n \ge 1$, the transducer has in memory only the value of $v^{(1)}_{n-1} \bmod b$.
As it reads $w_n$, for all $i$ it records $k^{(i)}_n \bmod b$ using $b$ states, from which it computes $\floor{v^{(i)}_n}$ using \Cref{eq::base-expansions-1} and $\gamma_i(n)$ using \Cref{eq::base-expansions--1}.

{\bfseries$\Acc{\beta}$ reduces to $\Acc{\gamma}$.} 
We will apply \Cref{thm::fancy-transduction-taylor's-version}.
The factorisation of $\beta$ is given by $\beta = \prod_{n=0}^\infty w_n$ as above.
Let $\Sigma_1 = \{0,\ldots,b-1\}^d$, $\Sigma_2 = \{0,1\}^{d+1}$, and $h \st \Sigma^*_2 \to M$ be a morphism into  a finite monoid.
We have to give a transducer~$\Bcal$ such that $\Bcal(\gamma) = \prod_{n=0}^\infty h(w_n)$.
By the classical lasso argument, from $M$ we can construct $L, m$ with $m \ge 1$ such that for any $x \in M$ and $n \ge L$, $x^n = x^{L + ((n-L) \bmod m)}$.

Using \Cref{eq::base-expansions-1.5} and \Cref{eq::base-expansions-3}, compute $N$ such that $k^{(i)}_n > L$ and $\floor{v^{(i)}_n} > L$ for all $i$ and $n \ge N$.
The transducer $\Bcal$ has the values of $h(w_n)$ hard-coded into it for all $n \le N$.
Before reading $\gamma(n)$ for $n > N$, it has in its memory the value of $r^{(i)}_{n-1} \coloneqq (\floor{v^{(i)}_{n-1}} - L) \bmod m$ for all $i$.
After reading $\gamma(n)$, it first computes $r^{(i)}_n \coloneqq (\floor{v^{(i)}_{n}} - L) \bmod m$ for all $i$ using \Cref{eq::base-expansions--2}.
Then it computes $x^{k^{(i)}_n}$ for all $i$, where $x = h((1,0,\ldots,0))$, using the fact that
\[
x^{k^{(i)}_n} = x^{	\lfloor v^{(i)}_n \rfloor - \lfloor v^{(i)}_{n-1} \rfloor - y_i} = x^{L + \big(\big(r^{(i)}_n - r^{(i)}_{n-1} - 1\big ) \bmod m\big)} 
\]
which follows from our construction of $N$.
Finally, it outputs $h(w_n) = x^{k^{(1)}_n} h(z_1) \cdots x^{k^{(d)}_n} h(z_d)$.
\end{proof}

\section{Discussion}
\label{sec::discussion}

The results of \Cref{sec::lrs-with-one-dominant-root} tell us everything that can be said about decidability of the MSO theories of linear recurrence sequences with one dominant root, barring major advances in the open decision problems of LRS including the Skolem Problem, the Positivity Problem, etc.
However, they do not really tell us anything new about the decidability of the MSO theory of $\langle \nat; <, \{u_n \ge 0 \st n \in \nat\}\rangle$, where $\seq{u_n}$ is an arbitrary integer LRS: our focus in this paper was rather on how to combine multiple predicates of arithmetic origin.
For an integer LRS $\seq{u_n}$ with a single dominant root (which must be real), decidability of the MSO theory of $\langle \nat; <, \{u_n \ge 0 \st n \in \nat\}\rangle$  can be shown using the approach of Carton and Thomas, or even Elgot and Rabin.
Recently, it was shown that for non-degenerate $\seq{u_n}$ with exactly two non-repeated non-real dominant roots, the MSO theory of $\langle \nat; <, \{u_n \ge 0 \st n \in \nat\}\rangle$ is decidable \cite{nieuwveld2025expansions}.
The proof is based on a novel idea: such $\seq{u_n}$ are \emph{pro-disjunctive}, defined as follows.
Order $\{u_n \ge 0 \st n \in \nat\}$ (which is guaranteed to be infinite due to the assumption on $\seq{u_n}$) as $\seq{p_n}$.
Then for any $m \ge 1$ and $w \in \Sigma^*$, where $\Sigma = \{0 \le r < m \st r\textrm{ occurs infinitely often in $\seq{p_n \bmod m}$}\}$, the word $w$ occurs infinitely often in $\seq{p_n \bmod m}$.
That is, rather than showing that $\seq{p_n \bmod m}$ is very structured for any $m\ge 1$, it is shown that $\seq{p_n \bmod m}$ is as random as possible for any $m\ge 1$.
However, for LRS with more than two dominant roots (as well as LRS with two repeated non-real dominant roots) decidability of the corresponding MSO theory remains open.

For predicates $P_1$ and $P_2$ of arithmetic origin, the following elementary property can be expressed in the monadic second-order (or even the first-order) language that is nevertheless of great interest.
There exist (infinitely many) pairs $n,m$ such that $n \in P_1$, $m\in P_2$, and $n-m = c$, where $c$ is a fixed integer.
For example, the famously open Brocard-Ramanujan problem is to determine whether $n! + 1 = m^2$ has any solution $(n,m) \in \nat^2$ with $n \notin \{4,5,7\}$.
Therefore, showing decidability of the MSO theory of $\langle \nat; <, \{n!\st n \in \nat\}, \kthpowers{2}\rangle$ would entail major mathematical breakthroughs.
Similarly, for any constant $k \ge 2$, the solutions $(n,m)$ of $|F_n - m^k| = 1$, where $F_n$ denotes the $n$th Fibonacci number, can be effectively enumerated \cite{bugeaud2008fibonacci}.
However, this is already highly non-trivial, and at the time of writing, no algorithm is known for enumerating all solutions  $(n,m)$ of $F_n - m^k = c$ for given constants $k \ge 2$, $c \in \intg$.
There are many other examples of arithmetic predicates whose MSO theories are connected to unsolved problems in number theory: in the cases we considered, the number-theoretic obstacles were all overcome using Baker's theorem.

\begin{acks}
Toghrul Karimov, Jo\"el Ouaknine, and Mihir Vahanwala are supported
by DFG grant 389792660 as part of \href{https://perspicuous-computing.science}{\texttt{TRR 248}}. Jo\"el~Ouaknine is also
supported by ERC grant DynAMiCs (101167561), and has a secondary affiliation with Keble College, Oxford as \href{https://www.emmy.network/}{\texttt{emmy.network}} Fellow.
Val\'erie Berth\'e is supported by  ERC grant DynAMiCs (101167561) and the Agence~Nationale de la Recherche through the project ``SymDynAr'' (ANR-23-CE40-0024-01).
James Worrell is supported by EPSRC Fellowship EP/X033813/1.

The Max Planck Institute for Software Systems is part of the Saarland
Informatics Campus.
This work partly arose through exchanges held at the workshop \emph{Algorithmic Aspects of Dynamical
Systems}, which took place in April and May 2023 at McGill
University's Bellairs Research Institute in Barbados.
\end{acks}

\bibliographystyle{ACM-Reference-Format}
\bibliography{refs}

\end{document}